\documentclass[sigconf]{acmart}

\setcopyright{rightsretained}

 \usepackage[subtle]{savetrees}

\usepackage{xspace}
\usepackage{amsmath}
\usepackage{amsthm}
\usepackage{amstext}
\usepackage{amsfonts}
\usepackage{tikz}
\usepackage{subcaption}

\usepackage{pgfplots}
\usepackage{pgfplotstable}
\usepgfplotslibrary{statistics}
\usepackage{tikz}
\usetikzlibrary{external}

\usepackage{algorithm}
\usepackage{algpseudocode}
\usepackage{wrapfig}
\usepackage{url}
\usepackage{enumerate}
\usepackage{enumitem}
\setitemize{leftmargin=*,noitemsep,topsep=1pt}
\usepackage[normalem]{ulem}
\usepackage{relate}
\usepackage{marginnote}
\usepackage{todonotes}
\usepackage{cleveref}
\usepackage{balance}
\usepackage{svg}

\usepackage{titlesec}

\renewcommand{\epsilon}{\varepsilon}

\newcommand{\defn}[1]		{{\textit{\textbf{\boldmath #1}}}\xspace}

\newtheorem{theorem}{Theorem}
\newtheorem{lemma}[theorem]{Lemma}

\newtheorem{proposition}[theorem]{Proposition}

\definecolor{magenta4}{rgb}{0.5625,0,0.5625}
\definecolor{green4}{rgb}{0,0.5625,0}
\definecolor{orange4}{rgb}{0.98,0.31,0.09}
\newcommand{\rob}[1]{{\scriptsize \textcolor{purple}{Rob: {#1}}}}
\newcommand{\mfc}[1]{{\scriptsize \textcolor{orange4}{Martin: {#1}}}}
\newcommand{\mab}[1]{{\scriptsize \textcolor{red}{Michael: {#1}}}}
\newcommand{\alex}[1]{{\scriptsize \textcolor{magenta}{Alex: {#1}}}}
\newcommand{\prashant}[1]{{\scriptsize \textcolor{blue}{Prashant: {#1}}}}
\newcommand{\richard}[1]{{\scriptsize \textcolor{blue}{Richard: {#1}}}}
\newcommand{\david}[1]{{\scriptsize \textcolor{blue}{David: {#1}}}}

\newcommand{\newtext}[1]{{{#1}}}

\newcommand{\depricate}[1]{}

\newcommand{\sysname}{\textsc{AdaptiveQF}\xspace}
\newcommand{\qf}{quotient filter\xspace}

\newcommand{\para}[1]{\smallskip\noindent\textbf{#1}.}

\newcommand{\etal}{\text{et al}.\xspace}

\newcommand{\calA}{\mathcal{A}}
\newcommand{\calF}{\mathcal{F}}
\newcommand{\calI}{\mathcal{I}}

\newcommand{\no}{\textsc{no}\xspace}
\newcommand{\yes}{\textsc{yes}\xspace}
\newcommand{\No}{\textsc{No}\xspace}
\newcommand{\Yes}{\textsc{Yes}\xspace}
\newcommand{\yesno}{\textsc{yes/no}\xspace}

\newcommand{\YesNo}{\textsc{Yes/No}\xspace}

\newcommand{\lcp}{\operatorname{\textsf{lcp}}}
\newcommand{\poly}{\operatorname{poly}}
\newcommand{\prob}[1]{\Pr\left[#1\right]}
\DeclareMathOperator*{\E}{\mathbb{E}}
\newcommand{\expect}[1]{\E\left[#1\right]}

\makeatletter
\algrenewcommand\ALG@beginalgorithmic{\footnotesize}
\makeatother

\algnewcommand{\algorithmicor}{\textbf{ or }}
\algnewcommand{\OR}{\algorithmicor}

\makeatletter
\newcommand{\InFloat}[2]{\@ifundefined{@captype}{#2}{#1}}
\makeatother

\title{Adaptive Quotient Filters}

\author{Richard Wen} 
\email{rwen1@umd.edu}
\affiliation{
  \institution{University of Maryland}
  \country{USA}
}

\author{Hunter McCoy} 
\email{hunter@cs.utah.edu}
\affiliation{
  \institution{University of Utah}
  \country{USA}
}

\author{David Tench} 
\email{dtench@pm.me}
\affiliation{
 \institution{Lawrence Berkeley National Labs}
  \country{USA}
}

\author{Guido Tagliavini} 
\email{guido.tag@rutgers.edu}
\affiliation{
 \institution{Rutgers University}
  \country{USA}
}

\author{Michael A.~Bender} 
\email{bender@cs.stonybrook.edu}
\affiliation{
 \institution{Stony Brook University}
  \country{USA}
}

\author{Alex Conway} 
\email{ajc473@cornell.edu }
\affiliation{
 \institution{Cornell Tech}
  \country{USA}
}

\author{Martin Farach-Colton} 
\email{martin@farach-colton.com}
\affiliation{
 \institution{New York University}
  \country{USA}
}

\author{Rob Johnson} 
\email{robj@vmware.com}
\affiliation{
 \institution{VMware Research}
  \country{USA}
}

\author{Prashant Pandey}
\email{ppandey@cs.utah.edu}
\affiliation{
 \institution{University of Utah}
  \country{USA}
}

\renewcommand{\shortauthors}{Richard Wen et al.}

\begin{document}


\begin{abstract}

Filters tradeoff accuracy for space and occasionally return false positive matches with a bounded error. Numerous systems use filters in fast memory to avoid performing expensive I/Os to slow storage. A fundamental limitation in traditional filters is that they do not change their representation upon seeing a false positive match. \newtext{Therefore, the maximum false positive rate is only guaranteed for a single and not an arbitrary set of queries.} We can improve the filter's performance on a stream of queries (especially on skewed distributions) if we can adapt after seeing false positives.


Adaptive filters, such as telescoping and adaptive cuckoo filters, update their representation upon detecting a false positive to avoid repeating the same error in the future. Adaptive filters require an auxiliary structure, typically much larger than the main filter and often residing on slow storage, to facilitate adaptation. 


However, existing adaptive filters are not practical and have seen no adoption in real-world systems due to two main reasons. Firstly, they offer weak adaptivity guarantees, meaning that fixing a new false positive can cause a previously fixed false positive to come back. Secondly, the sub-optimal design of the auxiliary structure results in adaptivity overheads so substantial that they can actually diminish the overall system performance compared to a traditional filter.



In this paper, we design and implement \sysname, the first practical adaptive filter with minimal adaptivity overhead and strong adaptivity guarantees, which means that the performance and false-positive guarantees continue to hold even for adversarial workloads. The \sysname is based on the state-of-the-art quotient filter design and preserves all the critical features of the quotient filter such as cache efficiency and mergeability. Furthermore, we employ a new auxiliary structure design which results in considerably low adaptivity overhead and makes the \sysname practical in real systems.

\newtext{We evaluate the \sysname by using it to filter queries to an on-disk B-tree database and find no negative impact on insert or query performance compared to traditional filters.} Against adversarial workloads, the \sysname preserves system performance, whereas traditional filters incur $2\times$ slowdown from adversaries representing as low as $1\%$ of the workload. Finally, we show that on skewed query workloads, the \sysname can reduce the false-positive rate $100\times$ using negligible (1/1000th of a bit per item) space overhead.



\end{abstract}


\maketitle

\pgfplotsset{
    legend entry/.initial=,
    every axis plot post/.code={%
        \pgfkeysgetvalue{/pgfplots/legend entry}\tempValue
        \ifx\tempValue\empty
            \pgfkeysalso{/pgfplots/forget plot}%
        \else
            \expandafter\addlegendentry\expandafter{\tempValue}%
        \fi
    },
}

\pgfplotsset{
  AQFStyle/.style         = {legend entry = AQF, color = cyan!50!blue,   fill = cyan!50!blue},
  TQFStyle/.style         = {legend entry = TQF, color = cyan,           fill = cyan},
  ACFStyle/.style         = {legend entry = ACF, color = violet!40!blue, fill = violet!40!blue},
  QFStyle/.style          = {legend entry = QF,  color = green!80!black, fill = green!80!black},
  CFStyle/.style          = {legend entry = CF,  color = orange,         fill = orange},
  SSStyle/.style          = {legend entry = SS,  color = purple,         fill = purple},
  CBFStyle/.style         = {legend entry = CBF, color = black,         fill = black},
  LineStyle/.style        = { fill=none, mark repeat = 2 },
  AQFLineStyle/.style     = {AQFStyle, mark = oplus*, LineStyle},
  TQFLineStyle/.style     = {TQFStyle, mark = square*, LineStyle},
  ACFLineStyle/.style     = {ACFStyle, mark = Mercedes star, LineStyle},
  QFLineStyle/.style      = {QFStyle,  mark = triangle*, LineStyle},
  CFLineStyle/.style      = {CFStyle,  mark = diamond*, LineStyle},
  SSLineStyle/.style      = {SSStyle,  mark = pentagon*, LineStyle},
  CBFLineStyle/.style     = {CBFStyle,  mark = diamond*, LineStyle},
  AQFBarStyle/.style      = {AQFStyle},
  TQFBarStyle/.style      = {TQFStyle},
  ACFBarStyle/.style      = {ACFStyle},
  QFBarStyle/.style       = {QFStyle},
  CFBarStyle/.style       = {CFStyle},
  SSBarStyle/.style       = {SSStyle},
  CBFBarStyle/.style      = {CBFStyle},
}

\newcommand{\throughputplotwidth}{2in}

\pgfplotsset{
  ThroughputBarGraph/.style = {
    ybar,
    width = \linewidth + 27pt,
    height = 115pt,
    bar width = 9pt,
    xtick = \empty,
    ymin = 0,
    error bars/error bar style = {color = black},
    error bars/y dir = both,
    error bars/y explicit,
    scaled y ticks = base 10:-6,
    ytick scale label code/.code={},
    yticklabel = { \pgfmathprintnumber{\tick}\hspace{0.08em}M },
    y tick label style = {rotate = 90},
    legend to name={insert-throughput-legend},
    legend columns = -1,
  },
  ThroughputTable/.style = {
    table/x = dummy,
    table/y = median,
    table/y error plus expr = {\thisrow{max} - \thisrow{median}},
    table/y error minus expr = {\thisrow{median} - \thisrow{min}},
  },
  SpaceTable/.style = {
    table/x = dummy,
    table/y = space,
  },
}

\newcommand{\workloadplotwidth}{2in}

\pgfplotsset{
  WorkloadGraph/.style = {
    width = \linewidth + 15pt,
    height = 110pt,
    ticks = major,
    try min ticks = 3,
    xlabel = Queries,
    enlarge x limits = false,
    scaled x ticks = base 10:-6,
    xtick scale label code/.code={},
    xticklabel = { \pgfmathprintnumber{\tick}\hspace{0.08em}M },
    y tick label style = {rotate = 90},
    y label style = {at = {(axis description cs:0.26,0.5)}, anchor=south},
    scaled y ticks = false,
    x label style = {at = {(axis description cs:0.5,0.1)}, anchor=north},
    xmin = 0,
    legend columns = -1,
    legend to name = trash
  },
  WorkloadTableFP/.style = {
    table/x = queries,
    table/y = fprate,
  },
  WorkloadTableSpace/.style = {
    table/x = queries,
    table/y = space,
  },
  HashAccessesTable/.style = {
    table/x = queries,
    table/y = accesses,
  },
}

\newcommand{\churnplotwidth}{3.3in}
\pgfplotsset{
  ChurnGraph/.style = {
    width = \linewidth + 18pt,
    height = 1.55in,
    ticks = major,
    try min ticks = 3,
    xlabel = Queries,
    enlarge x limits = false,
    scaled x ticks = base 10:-6,
    xtick scale label code/.code={},
    xticklabel = { \pgfmathprintnumber{\tick}\hspace{0.08em}M },
    y tick label style = {rotate = 90},
    y label style = {at = {(axis description cs:0.26,0.5)}, anchor=south},
    x label style = {at = {(axis description cs:0.5,0.1)}, anchor=north},
    ymin = 0,
    scaled y ticks = false,
    legend columns = -1,
    legend to name = trash
  },
  WorkloadTableFP/.style = {
    table/x = queries,
    table/y = fprate,
  }
}

\newcommand{\blacklistplotwidth}{3.3in}
\pgfplotsset{
  BlacklistGraph/.style = {
    width = \linewidth + 13pt,
    height = 1.55in,
    ticks = major,
    try min ticks = 3,
    xlabel = \No list/\Yes list Ratio,
    enlarge x limits = false,
    scaled x ticks = base 10:-6,
    xtick scale label code/.code={},
    xticklabel = { \pgfmathprintnumber{\tick} },
    scaled y ticks = false,
    y label style = {at = {(axis description cs:0.26,0.5)}, anchor=south},
    x label style = {at = {(axis description cs:0.5,0.1)}, anchor=north},
    legend columns = -1,
    legend to name = {blacklist-legend},
  },
  BlacklistTableFP/.style = {
    table/x expr = {\thisrow{ratio} / (1 - \thisrow{ratio})},
    table/y = fprate,
  },
  BlacklistTableSpace/.style = {
    table/x expr = {\thisrow{ratio} / (1 - \thisrow{ratio})},
    table/y expr = {\thisrow{space} + 200},
  }
}

\newcommand{\diskinsertplotwidth}{3.3in}
\pgfplotsset{
  DiskInsertGraph/.style = {
    width = \linewidth + 13pt,
    height = 1.55in,
    ticks = major,
    try min ticks = 3,
    xlabel = load factor,
    enlarge x limits = false,
    xtick scale label code/.code={},
    xticklabel = { \pgfmathprintnumber{\tick} },
    scaled y ticks = false,
    y label style = {at = {(axis description cs:0.26,0.5)}, anchor=south},
    x label style = {at = {(axis description cs:0.5,0.1)}, anchor=north},
    legend columns = -1,
    legend to name = {disk-insert-legend},
  },
  DiskInsertTableInserts/.style = {
    table/x = fill,
    table/y = through,
  },
  DiskInsertTableQueries/.style = {
    table/x = queries,
    table/y = through,
  },
  DiskInsertTableFPRates/.style = {
    table/x = queries,
    table/y = fprate,
  },
  DiskInsertTableMapAccesses/.style = {
    table/x = fill,
    table/y = accesses,
  },
}

\newcommand{\diskqueryplotwidth}{3.3in}
\pgfplotsset{
  DiskQueryGraph/.style = {
    width = \linewidth + 13pt,
    height = 1.55in,
    ticks = major,
    try min ticks = 3,
    xlabel = throughput (millions of ops/sec),
    enlarge x limits = false,
    scaled x ticks = base 10:-6,
    xtick scale label code/.code={},
    xticklabel = { \pgfmathprintnumber{\tick} },
    scaled y ticks = false,
    y label style = {at = {(axis description cs:0.26,0.5)}, anchor=south},
    x label style = {at = {(axis description cs:0.5,0.1)}, anchor=north},
    legend columns = -1,
    legend to name = {disk-query-legend},
  },
  DiskQueryTableHistogram/.style = {
    table/x = fill,
    table/y = through,
  }
}

\newcommand{\diskadversarialplotwidth}{3.3in}
\pgfplotsset{
  DiskAdversarialGraph/.style = {
    width = \linewidth + 13pt,
    height = 1.55in,
    ticks = major,
    try min ticks = 3,
    xlabel = adversarial,
    ylabel = throughput,
    enlarge x limits = false,
    xtick scale label code/.code={},
    scaled x ticks = base 10:2,
    xticklabel = { \pgfmathprintnumber{\tick} },
    scaled y ticks = false,
    y label style = {at = {(axis description cs:0.26,0.5)}, anchor=south},
    x label style = {at = {(axis description cs:0.5,0.1)}, anchor=north},
    legend columns = -1,
    legend to name = {disk-adversarial-legend},
  },
  AdversarialQueries/.style = {
    table/x = freq,
    table/y = through,
  }
}

\pgfplotsset{
  HighLoadThroughputGraph/.style = {
    ybar,
    width = \linewidth + 27pt,
    height = 115pt,
    bar width = 9pt,
    xtick = \empty,
    ymin = 0,
    error bars/error bar style = {color = black},
    error bars/y dir = both,
    error bars/y explicit,
    ytick scale label code/.code={},
    y tick label style = {rotate = 90},
  },
  HighLoadTable/.style = {
    table/x = dummy,
    table/y = through,
  },
}



\section{Introduction}\label{sec:intro}



\defn{Filters}~\cite{Bloom70,PandeyBJP17,FanAnKa14,GrafLe20} are a go-to data structure in systems builders' toolkits. Filters maintain a compact representation of a set of items, saving space by allowing a small \defn{false-positive rate $\epsilon$}:  a membership query to a filter for set $S$ returns \yes for any $x\in S$ and returns \no with probability at least $1-\epsilon$ for any $x\not\in S$.

Filters are powerful because allowing false-positives dramatically reduces the space required to store $S$. For example, if we are required to answer queries on $S$ with no errors, then the size of a data structure is at least $\log {u \choose n}= \Omega(n \log (u/n))$ bits, where $u$ is the size of the universe~\cite{CarterFG78}. In contrast, modern filters have size $n\log(1/\varepsilon) + cn$, where $c$ is between 2 and 3~\cite{bender2012don,PandeyBJP17}.  This means that, for typical false-positive rates around 1\% to 0.1\%, a filter can store one or two bytes of information per item, no matter how large the universe.  This bound is tight up to lower-order terms in that any filter requires at least $n\log(1/\varepsilon)$ bits~\cite{CarterFG78}.
%
Filters have been extensively used to compactly summarize a set of items in networks, storage systems, machine learning, computational biology, and other areas~\cite{BroderMi04, WangSuJi14, AlsubaieeBeBo14, ZhuLiPa08, DebnathSeLi11, DebnathDeLi10, ReagenGA17, jackman2017abyss,pell2012scaling, bradley2019ultrafast, solomon2016fast, chu2014biobloom, stranneheim2010classification,EsmetBeFa12, Farach-ColtonFeMo09,YuanZhJa16,JannenYuZh15a,JannenYuZh15b,yuan2017writes}.
%


\para{Types of problems}
The following problem settings offer challenges for traditional filters and opportunities for improvement:
\begin{itemize}
    \item \textbf{Static \yesno lists.}  Given a set $Y$ of \yes items and a set $N$ of \no items chosen from a universe $U$, build a data structure that answers \yes to any query for an item in $Y$, \no for any query for an item in $N$, and answers \no with probability at least $1-\epsilon$ for any other item in $U$.
    \item \textbf{Dynamic \yesno lists.} This is similar to the static \yesno list problem, except that the sets $Y$ and $N$ may be updated dynamically.
    \item \textbf{Skewed query distributions.} In some settings, the frequency distribution of queries may be highly skewed.
    In such settings, the observed false-positive rate of the filter can be very far from the expected rate, $\epsilon$.
    For example, if all the queries are for a single item, the observed false-positive rate will be 0 or 1, but not in between.
    Avoiding repeated mistakes can reduce the false-positive rate of a filter, or equivalently, reduce the filter size needed to achieve some target error rate.  In summary, filters that ignore the skew may perform arbitrarily poorly, whereas filters that exploit the skew can outperform the lower bounds.
    \item \textbf{Adversarial queries.}  In this problem, the goal is to design a filter that  guarantees that the fraction of queries from an adversary that are false positives is at most $\epsilon$, even when the queries are chosen by an adversary that is trying to cause the filter to return as many false positives as possible.  Here we assume the attacker can detect when a query results in a false positive and can repeat queries arbitrarily.  This is a more general case of the skewed-query distribution.
\end{itemize}

\para{Prior work}
Prior work has considered each of these problems separately, and consequently has developed distinct approaches to solving each of them.  Chazelle et al.~\cite{chazelle2004bloomier} describe \defn{Bloomier filters}, which encode static \yesno lists.  Bloomier filters also support a limited form of dynamicity---they support moving items between $Y$ and $N$ but not adding or deleting items.
Tripunitara and Carbunar~\cite{TripunitaraCa09} introduce \defn{cascading Bloom filters} to solve the static \yesno list problem, and these are used in many  systems~\cite{salikhov2014using,ChikhiRi13,RozovShHa14,MousaviTr19,Larisch2017}. Reviriego et al.~\cite{reviriego2021} proposed an extension of the \defn{static xor filter} to support the static \yesno list problem.  Li et al.~\cite{seesaw2022} proposed the \defn{seesaw counting filter} for the dynamic \yesno list problem, specifically in the context of detecting malicious URLs.
Mitzenmacher et al.~\cite{MitzenmacherPR20} proposed \defn{adaptive cuckoo filters} to solve the skewed-query-distribution problem.
Bender et al.~\cite{BenderFaGo18} define the notion of an adaptive filter, which offers strong guarantees on the number of false positives that an application will see, even with a skewed or even adversarial query distribution, and present the broom filter, which meets their definition.  Bender et al.~\cite{BenderDaFa21} analyzed the performance of \defn{broom filters}~\cite{BenderFaGo18} on queries that obey Zipfian distributions. Lee et al.~\cite{lee2021telescoping} proposed \defn{telescoping filters} to address the skewed query distribution problem.


\para{This paper} We argue that all of these problems can be naturally solved, with comparable space and better performance than the prior special-purpose solutions, by what we call a \defn{monotonically adaptive} filter, which is a filter that never forgets a false positive.  Furthermore, we show how to build fully dynamic monotonically adaptive filters from quotient filters~\cite{PandeyBJP17a}.  We design, build, and evaluate a fully dynamic monotonically adaptive filter, the \sysname, and show that it outperforms several prior solutions to these problems. Finally, we prove lower bounds on the space required to solve the \yesno list problem, showing that the \sysname is space-optimal.

Like adaptive filters, monotonically adaptive filters can \defn{adapt}, i.e.~they can update their state to correct false positives. Bender \etal defined what it means for a filter to be adaptive: every query has a probability of at most $\epsilon$ of returning a false positive, independent of the outcome of all prior queries~\cite{BenderFaGo18}.  Adaptivity is a very strong property: it guarantees that after $n$ queries---even adversarially generated queries---the upper bound on the number of false positives is tightly concentrated around $\epsilon n$. Specifically, the system will see at most $\epsilon n + O\big(\sqrt{ \epsilon n\log n}+\log n\big)$ false positives with high probability. 

\textbf{Adaptive filters require two things: feedback about their false positives and auxiliary data to correct them.}  For example, if an adaptive filter is used by an application to avoid database lookups for non-existent items, then the application can inform the filter that a query for an item $x$ was a false positive if the subsequent database query returned that $x$ is not present in the database.  Adaptive filters also need an auxiliary structure to store information to support adaptation.  Bender \etal showed that this auxiliary information is necessary and in fact must be quite large: the total size of an adaptive filter on a set $S$ essentially must be large enough to store $S$~\cite{BenderFaGo18}.  
The trick is to break the filter into two parts, a small in-memory component that is accessed on every query and a large auxiliary structure that is accessed only during adaptations and hence can reside in slower storage.  Note that all proposed adaptive filters have this overall structure.  In some applications, such as when the filter is in front of a database, the database may be able to serve as the auxiliary structure, so that the total storage requirements of the system remain essentially unchanged. See Bender \etal for more discussion~\cite{BenderFaGo18}.

What makes monotonically adaptive filters special is that, when they adapt, their false-positive set only shrinks. Prior proposed adaptive filters were not monotonic: fixing one false positive could cause other elements to become false positives. Even filters that meet Bender \etal's strong definition of adaptivity need not be monotonic. For example, Bender \etal's broom filter periodically rotates its hash function, at which point it forgets all the false positives it corrected under the old hash function.

\defn{Fingerprint filters}, such as the quotient filter~\cite{PandeyBJP17a}, are good candidates for building practical monotonically adaptive filters because they store a set $S$ by compactly storing the set $h(S)=\{h(x) \mid x\in S\}$, where $h$ is a hash function and $h(x)$ is called the \defn{fingerprint} of $x$.  A query for $y$ simply checks whether $h(y)\in h(S)$, so the only source of false positives is fingerprint collisions. Fingerprint filters support a false-positive rate of $\epsilon$ on a set of size $n$ by using $\log (n/\epsilon)$-bit fingerprints and typically store the first $\log n$ bits of each fingerprint implicitly so that the per-item space is $\log (1/\epsilon) + O(1)$ bits. To make a fingerprint filter monotonically adaptive, we need only to be able to eliminate fingerprint collisions. To do so, we can use a hash function $h$ that outputs a large number of bits and initially store the first $\log (n / \epsilon)$ bits of $h(x)$ for each $x\in S$, where $n=|S|$. Whenever we discover a false positive, i.e.~a query $y$ whose fingerprint matches a fingerprint for some $x\in S$, we modify the filter to store a longer fingerprint for $x$ until the collision disappears. In fact, Kopelowitz et al.~\cite{kopelowitz2021support} show that this approach is not only natural but necessary: a space-efficient fingerprint filter must have variable-length fingerprints to be adaptive.  

We implement a prototype fully dynamic monotonically adaptive filter, the \sysname, that uses only $(1 + o(1))n\log(1/\epsilon) + O(n)$ space. We demonstrate experimentally that it outperforms existing purpose-built solutions for skew distribution and \yesno list workloads in terms of space efficiency, insertion speed, and query speed. We also show a space lower bound for the static \yesno list problem and that the \sysname meets the space lower bound up to low-order terms.

The challenge is to store and update these variable-length fingerprints efficiently in terms of space and time. 
Prior theoretical proposals for building adaptive filter have had complex mechanisms for managing variable-sized fingerprints~\cite{BenderFaGo18}.  We propose a simple scheme for implementing variable-sized fingerprints within \sysname.  Even though adapting requires extending a fingerprint by only two bits in expectation~\cite{BenderFaGo18}, \sysname simplifies fingerprint management by \emph{over-adapting}, i.e.~fingerprints grow by multiples of $\log (1/\epsilon)$ bits.  Over-adaptation could cause the filter to use too much space and over-minimize the false-positive probability below $\epsilon$. 
However, in practice, this is not an issue. 

\para{Our results} We evaluated the \sysname in isolation and as a component of larger systems.  The high-level summary of our findings is that the \sysname can speed up query throughput by delivering far fewer false positives than non-adaptive filters. We compared the performance of the \sysname to that of two other adaptive filters, the telescoping quotient filter (TQF)~\cite{lee2021telescoping} and the adaptive cuckoo filter (ACF)~\cite{MitzenmacherPR20}. We also compared it to two non-adaptive filters, the quotient filter (QF)~\cite{PandeyBJP17} and the cuckoo filter (CF)~\cite{FanAnKa14}.

\begin{enumerate}[leftmargin=*]
    \item In a disk-based database, \sysname is between 10$\times$ --- 30$\times$ faster than other adaptive filters (TQF, ACF) for overall insertion performance and is comparable to non-adaptive filters.
    \item In a disk-based database, \sysname achieves between 15\% --- 6$\times$ faster overall query performance than non-adaptive filters (QF, CF) for adversarial queries and has comparable performance for uniform-random query workloads.
    \item \sysname is dynamic (i.e.~support deletes and resizability) but still achieves similar or better space usage compared to purpose-built solutions for the static \yesno list problem.
    \item \sysname has negligible adaptivity overhead compared to the quotient filter on which is based.
    \item \sysname preserves all the critical features of the quotient filter such as mergeability, resizeability, and bulk insertions.
\end{enumerate}

\textbf{In summary, the adaptivity overhead is minimal in \sysname compared to non-adaptive filters. It is able to substantially improve overall system performance in scenarios where disk accesses incur a large cost. Furthermore, it matches or beats the performance of custom-built solutions for static \yesno list problems.}






\section{Filters and Applications}

In this section, we give an overview of general-purpose filters and adaptive filters. We then describe applications that can benefit from adaptive filters and review existing purpose-built filters for these applications. We divide these applications in two broad categories:

\begin{enumerate}[leftmargin=*]
    \item Applications that use traditional filters but have skewed workload patterns where adaptive filters can help.
    \item Applications that currently use purpose-built solutions where the cost of certain false positives is very high.
\end{enumerate}

\subsection{General-purpose filters}

For decades, the Bloom filter~\cite{Bloom70} was essentially the only available filter, but Bloom filters are suboptimal in terms of space usage, running time, and data locality, and they support a bare-bones set of operations (insert and lookup).   The Bloom filter has inspired numerous variants~\cite{QiaoLiCh14, PutzeSaSi07, LuDeDu11, CanimMiBh10, DebnathSeLi11, AlmeidaBaPr07, FanCaAl00, BonomiMiPa06}. The counting Bloom filter (CBF)~\cite{FanCaAl00} supports deletes at the cost of space.  
%
%
The blocked Bloom filter~\cite{PutzeSaSi07} provides better cache locality than the standard Bloom filter but it comes at a higher false-positive rate.

The quotient filter (QF)~\cite{PaghPR05,DillingerMa09,BenderFaJo12a,EinzigerFr16,PandeyBJP17} is a fingerprint filter. It stores fingerprints using Robin Hood hashing~\cite{CelisLaMu85}.
It divides the fingerprint into two parts, higher order $\log(n)$-bits as the quotient and lower order $\log(1/\epsilon)$-bits as the remainder. The quotient bits are used to locate a slot in the table, and the remainder bits are stored in that slot. 
It supports insertion, deletion, lookups, enumeration, resizing, and merging. The counting quotient filter (CQF)~\cite{PandeyBJP17}, improves upon the performance of the quotient filter and adds variable-sized counters to count items using asymptotically optimal space, even in large and skewed datasets.

The cuckoo filter~\cite{FanAnKa14} also stores fingerprints but uses cuckoo hashing instead of Robin Hood hashing. The Morton filter~\cite{BreslowJ18} is a variant of the cuckoo filter that is designed to speed up insertion using optimizations designed for hierarchical-memory systems.

\subsection{Strongly adaptive filters}
A \emph{strongly adaptive filter} modifies its state so that if a false positive is repeated, the probability that it is still a false positive is at most $\epsilon$. 
%
%
Bender et al.~\cite{BenderFaGo18} introduce the broom filter and the notion of strong adaptivity used in this paper. 
The broom filter is based on the quotient filter, but supports variable-length fingerprints to adapt to (and correct) false positives.
Lee et al.~\cite{lee2021telescoping} introduce the telescoping filter, which is also built using the quotient filter~\cite{PandeyBJP17}. The telescoping filter is strongly adaptive but avoids directly extending fingerprints. They change the remainder (or the tag) stored in the filter by using a different lower-order $\log(1/\epsilon)$-bits. Additionally, they maintain a table to record which $\log(1/\epsilon)$-bits they have used as a remainder in the filter for the adapted items. As the filter adapts, the size of this table grows (so in this sense their fingerprints are variable-length and strong adaptivity is possible). The adaptive cuckoo filter of Mitzenmacher et al.~\cite{MitzenmacherPR20} is adaptive in the sense that it changes its representation in response to false positives, but is not strongly adaptive. See Section~\ref{subsec:skew}.


\subsection{Filters for skewed query distributions}
\label{subsec:skew}
A common and important application of filters is their use in key-value stores based on Log-Structured Merge Trees (LSMs) and $B^\varepsilon$-trees~\cite{OneilChGa96,BrodalFa03b}. In these key-value stores, filters are used to avoid performing multiple expensive disk accesses per query~\cite{RocksDB,LevelDB,ConwayGC20,ChangDG06,WiredTiger,ScyllaDB,Cassandra}. In some database systems, this type of key-value store is used as the storage engine~\cite{RocksDB,LevelDB,ConwayGC20,Cassandra}.

In an LSM tree, data is stored in SSTables, which are static, sorted arrays of key-value pairs. The SSTables are organized into levels L0, L1, \ldots, where L0 is the smallest and holds the most recently written data. Each subsequent level is larger by a factor of $g$, where $g$ is a configuration parameter. As data is written, SSTables are moved down the structure and are merged into each other according to a chosen compaction policy. In general, at any point in time, a given key can be present in an SSTable on any level (even multiple SSTables per level in some variants), so queries need to check at least one SSTable on each level, which is expensive.
As a result, in almost all practical systems, each SSTable has a corresponding filter so that queries only read data from the SSTable when the key is present there or due to rare false positives from the filter.  Note that in this application, queries to SSTables on smaller levels are often negative since most of the data is stored in the larger levels. However, the smaller levels may contain recent updates, so they cannot be skipped. Thus, filters in the smaller levels see frequent negative queries~\cite{dayan2017monkey}.  

One challenge to applying adaptive filters to LSMs is that the SSTables are typically static and most LSM trees store their filters in their SSTables. However, this is not necessary or universal.  For example, SplinterDB stores its filters separately from the data they cover~\cite{ConwayGC20}.  An LSM could even store adaptive versions of its filters only in memory.  The adaptivity information would be lost on a crash, meaning that false positives might increase after a crash, but this should be rare enough to be insignificant.

Query workloads to database systems commonly follow a power-law or otherwise skewed distribution~\cite{Newman05,ClausetSh09}.  As a result, state-of-the-art benchmark suites YCSB, TPC-C, and TPC-E incorporate data skew into most of their workloads~\cite{CooperSi10,LeuteneggerDi93,Hogan09,ChenAA11}. In fact, many systems have tried to mitigate or even exploit the effects of data skew~\cite{NgoRe14,ChengKo14,BeameKo14,DugganPa15,ZhangRo22}.
A skewed database query workload also results in a skewed workload for filters. Moreover, as noted above, many of the queries will be negative. Adaptive filters, such as \sysname, can outperform non-adaptive filters on skewed workloads by eliminating repeated false positives on frequently accessed keys.

\subsection{Filters for \YesNo list problems}

\para{Detecting Malicious URLs} 
Malicious websites pose a major threat to internet users. For example, merely visiting a malicious URL may cause a user's web browser to be hijacked~\cite{sun2016automating}. 
%
Since URLs are long~\cite{urlcount} and abundant~\cite{kaspersky2021}, an effective way for a router to block malicious URLs is to store them as the \yes list of a filter~\cite{seesaw2022}. 

However, it is important not to block legitimate URls that are false positives~\cite{deeds2020stacked}, so every positive response of the malicious-URL filter must be verified~\cite{nytwait2012,leng_2022}, which is expensive.  
This additional overhead imposed on false positive (safe) URLs is especially undesirable when the URL is \textit{important}. For instance, a false positive may block access to a voter registration webpage, or emergency weather information, whereas slowing the loading of other false-positive pages is relatively benign. 


One way to address this variability in false positive cost is to store important false positives in a \no list, so that they are never blocked and so they do not pay the URL-verification penalty.   Chazellete et al.~\cite{chazelle2004bloomier} introduced the Bloomier filter which solves the \yesno list problem. Li et al.,~\cite{seesaw2022} present the Seesaw Counting Filter (SSCF), which implements a \yesno list filter specifically for the malicious URL blocking problem.
Reviriego et al.~\cite{reviriego2021} present the Integrated Filter which also implements a \no list. Both focus on the case where the \no list is static and known ahead of time. 
The SSCF has an extension for adding \no list items dynamically, but it is not guaranteed to prevent false positives by doing so and can also introduce false negatives.

URL requests may also vary in frequency, and these frequencies may even change over time. However, the existing filters literature on this problem does not consider this generalization. For this work we restrict our focus to the standard assumption that a static set of high-priority elements must never be false positives.

\para{Certificate Revocation Lists}
In the TLS PKI (Transport Layer Security Public Key Infrastructure~\cite{housley1999internet}), browsers should check whether a certificate has been revoked before trusting connections authenticated by the certificate.  Traditionally this was done via a ``pull'' approach, i.e., browsers would check with a central repository of revoked certificates when they established a connection.  More recent work has sought to move to a ``push'' model, where browsers receive frequent updates to the list of revoked certificates, so the browser can perform a purely local check when it establishes a new connection.

Larisch, et al.,~\cite{Larisch2017} proposed CRLite, which uses cascading Bloom filters to store the set of revoked certificates at the client. They observed that, in the case of TLS certificates, the universe is a small finite set and known at construction time. They can build a cascade of Bloom filters to exactly represent the set of revoked certificates. In the cascading Bloom filter, each subsequent Bloom filter contains false positive set from the earlier Bloom filter until the false positive set is small enough to be stored exactly in a hash table. A central system would periodically push updates to this list to browsers.  The updates are encoded as bitwise deltas on the original filters.  When the space of certificates grows too large, so that they need to resize the filters, then they have to transmit new filters from scratch.

\para{De Bruijn graph traversal}  In computational biology, de Bruijn graphs (DBGs) are at the heart of numerous genomic sequence analysis pipelines~\cite{PandeyBJP17a, PandeyBJP17b}. In a de Bruijn graph, each node is a $k$-length subsequence (of the DNA bases, ``A'', ``C'', ``G'', and ``T'') from the underlying biological samples, and two nodes are connected via an edge if they share a $(k-1)$-length subsequence.  Analyses traverse DBGs during assembly, error correction, ``contig'' detection, and numerous other applications.

De Bruijn graphs are often large enough that they do not fit in the memory.  Numerous methods have been proposed to exploit their special structure for compression.  One of the main tricks is to take advantage of the fact that each node has at most 4 incoming edges and 4 outgoing edges (one for each base that can be prepended or appended to the node).  Thus a traversal can query for the existence of each edge, so we can represent the DBG using an (approximate) set data structure that supports only membership queries, i.e. a filter.  In this application, false positives in the filter result in extra edges in the graph.
To avoid the false edges, Chikhi and Rizk~\cite{chikhi2013space} proposed to store the de Bruijn graph in a cascading Bloom filter as the set of queries is known in advance. Each Bloom filter stores the false positives from querying the earlier Bloom filter using all possible queries during the dbg traversal. 

\section{\sysname design}
\label{sec:design}

In this section, we describe the high-level scheme of our solution, without worrying about how to encode this design in a small number of bits.  The encoding is described in subsequent sections.

\begin{figure}[t]
    \centering
    \includesvg[width=\columnwidth]{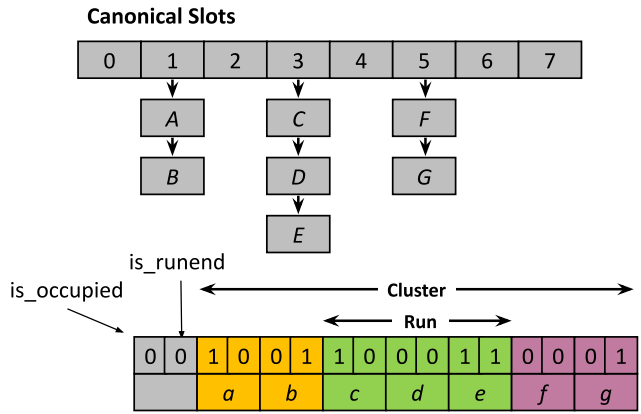}
    \caption{The quotient filter~\cite{PandeyBJP17} structure. The upper part shows the logical structure. The lower part shows the encoding of the logical structure in the quotient filter.
    It uses two metadata bits per slot. All items that share the same canonical location are stored together in a {\it \bf \em run}. A sequence of items without any empty slot is called a {\it \bf \em  cluster}. Note: the items are showed in upper case in the canonical representation and the remainders corresponding to the items in the slots are showed in lower case.}
    \label{fig:qf_structure}
\end{figure}

\subsection{High-level design}

The \sysname builds on the idea of fingerprint filters, which store a set $S$ by storing the set of \defn{fingerprints} $h(S)=\left\{h(x)\mid x\in S\right\}$.  The basic idea behind the \sysname is that, initially, we store only enough bits of each fingerprint to ensure a false positive rate of $\epsilon$, i.e.~we store a set $F$ of fingerprints, where each fingerprint is actually a \textit{prefix} of $h(x)$, for some $x\in S$.   A query for $y$ returns \yes if some fingerprint in $F$ is a prefix of $h(y)$.
When we discover a false positive, i.e., an item $y\not\in S$ such that some fingerprint $f\in F$ is a prefix of $h(y)$, we increase the length of $f$ until it is no longer a prefix of $h(y)$.

The main issue that arises is: how can we extend a fingerprint $f$ in our filter without knowing the full hash $h(x)$ of which it is a prefix?  To solve this problem, all adaptive filters maintain a reverse map that maps fingerprints in $F$ back to their full hashes (or even the original keys).  Notably, the need to store full hashes means this map will be much larger than the filter, possibly too large to fit in fast storage. Thus, we would like to minimize how often this reverse map needs to be updated and/or queried. 

At the very minimum, one insert needs to be done to the reverse map for each insert to the filter. In addition, one reverse map query needs to be made for each adaptation. However, since adaptations are responses to false positives, a disk access is done at this point anyway, and adapting ensures that the offending query will not cause another disk access in the future. Thus, the \sysname maintains its fast performance on general queries. Ideally, we would access the reverse map at no other time.



In the following sections, we describe how we store and update variable-length fingerprints efficiently and how we maintain the reverse map will little overhead.

\subsection{Quotient filter}\label{sec:qf-design}

We build the \sysname using the \qf~\cite{PandeyBJP17}. The \qf has the ability to associate small variable-length values with fingerprints. We exploit this feature to extend the fingerprint size to adapt. Using the \qf as the underlying filter helps retain advantages, such as good cache-locality, deletion, resizability, enumerability, mergeability, etc., that the \qf has over other filters. In this section, we give an overview of Pandey et al.'s \qf~\cite{PandeyBJP17}. Later in~\Cref{sec:implementation}, we explain how we modify the \qf schema to build the \sysname.


The quotient filter (QF) stores an approximation of a multiset $S\subseteq \mathcal{U}$ by maintaining a compact, lossless representation of the multiset $h(S)$, where $h:\mathcal{U}\rightarrow\{0,\ldots,2^p-1\}$ is a hash function that maps items from the universe $\mathcal{U}$ to a $p$-bit fingerprint. To handle a multiset of up to $n$ distinct items while maintaining a false-positive rate of at most $\epsilon$, the QF sets $p=\log_2 \frac{n}{\epsilon}$ (see~\cite{BenderFaJo12a}  for the analysis).

The \qf uses Robin-Hood hashing~\cite{CelisLaMu85} to store the fingerprints compactly in a table. It consists of an array $Q$ of $2^q$ \defn{slots} and a hash function $h$ mapping items from a multiset to $p$-bit integers, where $p \geq q$.
Robin-Hood hashing is a variant of linear probing in which we try to place an item $a$ in slot $h(a) / 2^{p-q}$, but shift items down when there are collisions to create empty space. Robin-Hood hashing maintains the invariant that, if $h(a) < h(a')$, then $a$ will be stored in an earlier slot than $a'$.

The \qf divides $h(x)$ into its first $q$ bits, \defn{quotient} $h_0(x)$, and its remaining $r$ bits, \defn{remainder} $h_1(x)$.  Together, the quotient and remainder form the \defn{fingerprint} of $x$.  The \qf maintains an array $Q$ of $2^q$ $r$-bit slots, each of which can hold a single remainder.  When an element $x$ is inserted, the \qf attempts to store the remainder $h_1(x)$ at index $h_0(x)$ in $Q$
(which we call $x$'s \defn{canonical slot}).  If that slot is already in
use, then the \qf uses Robin hood hashing to find the next available empty slot to store $h_1(x)$. All the items that share the same canonical slot are stored together in a \defn{run} and a sequence of runs stored contiguously with no empty space is called a \defn{cluster}. During an insert operation, the next available empty slot is found at the end of the cluster. If an item lands at the start of the cluster then all the items in cluster must be shifted to create an empty space (see~\Cref{fig:qf_structure}).

The \qf also maintains 2 bits of additional metadata ({\it is\_occupied} and {\it is\_runend}) per slot in order to determine which slots are in use and the canonical slot of each remainder stored in $Q$. When an item is inserted into the canonical slot, the {\it is\_occupied} bit for that slot is set to 1. The {\it is\_runend} bit is set to one for every slot that contains the last remainder in a run. 
Please refer to Pandey et al.~\cite{PandeyBJP17} for further details.

\begin{figure}
\centering
\begin{subfigure}{0.47\textwidth}
    \includesvg[width=\textwidth]{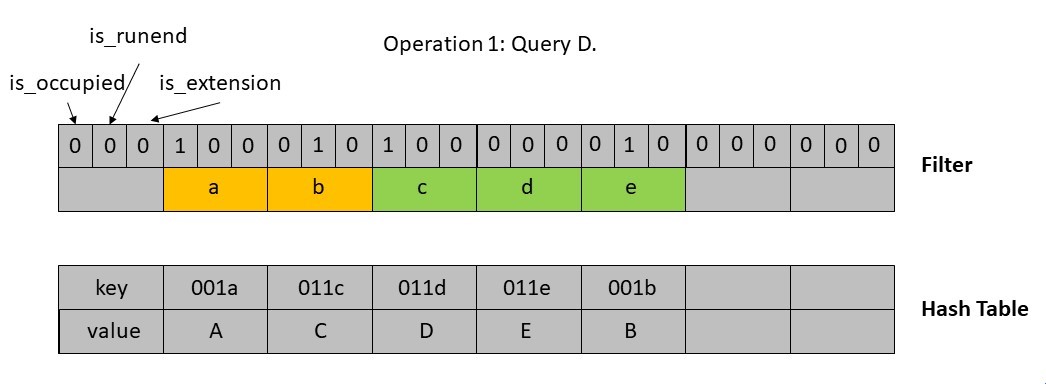}
    \caption{\emph{Querying D.} Hashing D gives 011d, and 011 is the quotient of the 4th bucket . The run begins at the 4th slot, and matching remainder d is found in this run in slot 5. Filter returns YES.}
    \label{fig:first}
\end{subfigure}
\hfill
\begin{subfigure}{0.47\textwidth}
    \includesvg[width=\textwidth]{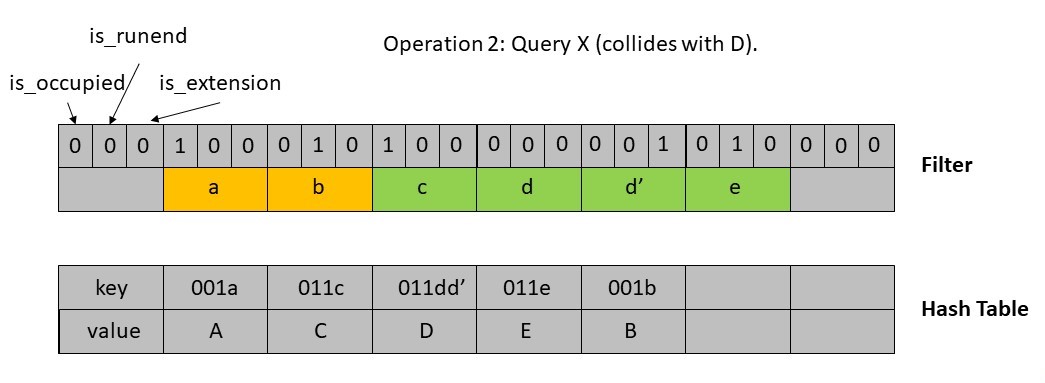}
    \caption{\emph{Querying X.} X hashes to 011d, so query is as (a). When informed that X is a false positive, note that the fingerprint responsible is the first fingerprint with quotient-remainder pair 011d. Thus, look up (011d, 0) in hash table and extend fingerprint to 011dd' according to D. Add extra slot for d' and set extension bit, shifting other slots right. Note that the reverse map does not need to be modified because the fingerprint is still the first fingerprint with quotient-remainder pair 011d.}
    \label{fig:second}
\end{subfigure}
\hfill
\begin{subfigure}{0.47\textwidth}
    \includesvg[width=\textwidth]{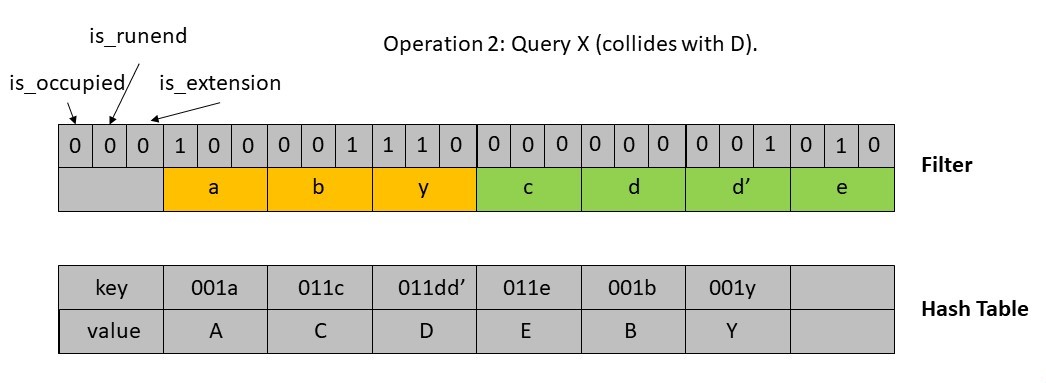}
    \caption{\emph{Inserting Y.} Y happens to hash to 001a. Find the run for bucket 001 as for query. Look for remainder $a$ in the run. Having found an existing minirun, add y to the end of the minirun by shifting everything to the right. Add Y to hash table associated with the second fingerprint with quotient-remainder pair 001a.}
    \label{fig:third}
\end{subfigure}
        
\caption{\sysname block diagram and the reverse map. It shows the changes in the schema and the reverse map during queries, insertions, and adaptations.}
\label{fig:aqf-diagram}
\end{figure}

\section{\sysname Implementation}\label{sec:implementation}

The AQF uses the same fundamental structure as the QF. Here we describe how we modify the QF schema to support adaptivity.  

\subsection{False positives}
First, let us understand how false positive queries occur in the QF. The QF stores the fingerprints compactly and exactly in the table. Therefore, a false positive occurs due to a hash collision while computing the fingerprint. That is, there exist two distinct items ($x$ and $y$) that share the same fingerprint. If $x \in S$ but $y \not\in S$, then a query for $y$ will result in a false positive.

A QF guarantees a false-positive rate~\footnote{The false-positive rate is defined as the ratio of the number of false positives reported over the total number of queries in the set.} of $\epsilon = 2^{-r}$, where $r$ is the number of remainder bits, for a set of items drawn uniformly at random from the universe $\mathcal{U}$.
%
However, if the items are drawn from some arbitrary distribution then there is no guarantee on the false positive rate. For example, if a query set consists of a single item that happens to be a false positive, then the false positive rate of the filter for that set will be 1.

\subsection{Adapting to false positives}
The \sysname updates its representation on every false positive query. These updates ensure that the false-positive probability for any query $x$ remains $\leq \epsilon$, even if $x$ was a false positive the last time that it was queried.

\sysname updates its representation by using an additional slot for the false positive item, extending its fingerprint by the next $r$ bits in its hash, known as an \defn{extension}. Multiple extensions can be added if the resulting fingerprint still produces a false positive on the given query, with the probability that the fingerprint incurs a false positive decreasing by a factor of $2^{-r}$ with each extension. The \defn{fingerprint} of an item in the \sysname now refers to its quotient, remainder, and any extensions that have been added to it in the filter.

We introduce an additional overhead bit, the \emph{is\_extension} bit, to differentiate between slots storing remainders and extensions, bringing the total overhead to 3.125 bits per slot. Recall that the QF has an overhead of 2.125 bits per item. Slots with an unmarked \emph{is\_extension} bit are treated as usual. A marked \emph{is\_extension} bit indicates that the slot contains an extension of the previous remainder.

\para{Lists in the reverse map}
To extend a fingerprint, we need to obtain the original key mapped to the current fingerprint in the filter. To do this, we maintain an on-disk hash table that functions as a reverse map from fingerprints to the keys. We extend the fingerprint for the mapped key by adding additional bits derived from the original hash of the key. Recall that the reverse map is not specific to the \sysname and is needed by any adaptive filter to retrieve additional information necessary for adaptation.


However, because the \sysname adapts by appending to fingerprints rather than reshuffling them as done in the telescoping filter~\cite{lee2021telescoping} and the adaptive cuckoo filter~\cite{MitzenmacherPR20}, the existing bits of the fingerprint prior to adaptation stay the same. Most notably, the quotient and remainder, which every fingerprint is guaranteed to have, are fixed from the moment of insertion. We use this fact to construct a reverse map that does not need to be accessed in response to natural shifting in the filter during insertions.

Suppose two items $x$ and $y$ share a quotient and remainder. Because runs are sorted by quotient, and fingerprints within a run are sorted by remainder, it follows that the fingerprints of $x$, $y$, and any other items with the same quotient-remainder pair are stored contiguously in the \sysname. Let us call this group of fingerprints a \defn{minirun}. We will call these fingerprints' shared quotient-remainder pair their \defn{minirun ID}. Note that even when an item does not share a quotient-remainder pair with any other item, its fingerprint is still contained in its own minirun of length 1. The \defn{minirun rank} of an item is the rank of its fingerprint within its minirun.
%
The reverse map will map from a minirun ID to a list of all the inserted keys with that minirun ID. We order the keys in the list according to the order of their fingerprints in the minirun. 


To insert an item into the database, we start by inserting its fingerprint into the filter. We locate the slot in which the fingerprint belongs, inserting it into the filter and obtaining its minirun rank in the process. We then turn to the reverse map and attempt to map that item's minirun ID to the item itself. If this is the first instance of its minirun ID, we create a new linked list consisting of only that item. If a previously-inserted item shares the same minirun ID, we insert the new item into the existing linked list in the position indicated by its minirun rank. As a result, all subsequent items are shifted back one position in the linked list, matching how their associated fingerprints in the minirun also shifted back one position when we made the filter insertion.

When querying for an item in the database, we always start by querying the filter for its fingerprint, obtaining its minirun rank if the fingerprint is present. We then query the reverse map for the minirun ID, retrieving the associated linked list and obtaining the item at the position indicated by the minirun rank. We can detect a false positive because the original key located at that position in the linked list will differ from the queried key. In that case, we use the original key we found to adapt the offending fingerprint in the filter, which we can do without modifying the reverse map.

\Cref{fig:aqf-diagram} shows the representation of the \sysname and the reverse hash map during insertions and adaptations during queries.

\para{The reverse map as a database} 
So far, we have acted as though the reverse map and the database are two separate data structures, both on disk. As a result, the reverse map would seem to incur additional disk accesses for every insertion and positive query. We call this the \defn{split} reverse map setup. However, we can merge the database and reverse map into a single key-value store, eliminating this cost. Because the reverse map allows us to uniquely identify the original key of any fingerprint in the filter, we can simply store any relevant values next to these original keys in the reverse map. That is, instead of the database and reverse map being two separate mappings (one from keys to values and the other from fingerprints to keys), we use the reverse map as a single mapping from fingerprints to keys and values. In effect, the reverse map is being used as a replacement for the database. This does not require any additional queries over the conventional key to values mapping -- the filter is always queried before the database, so we already obtain the fingerprint for any queried item. We call this the \defn{merged} reverse map setup. For all of our disk experiments, we use the merged setup. However, note that this optimization leads to the items in the database being stored in hash order. Thus, this would no longer support range queries. For applications that wish to use range queries, the split reverse map would suffice. We evaluate the overhead of using the split setup later. 


\para{Counters}
Like the counting quotient filter upon which it is based, the adaptive quotient filter supports storing multiple copies of the same item by storing a variable-length counter with each fingerprint.  In the CQF, each item had one slot holding its remainder, followed by 0 or more slots that encoded the number of times that fingerprint was present in the filter.  The CQF encoded the counter such that singleton fingerprints used zero additional slots for their count, which meant the CQF was just as space efficient as a non-counting quotient filter when the set contained no duplicates.  However, because the CQF has only two metadata bits, encoding the counters was quite complex.  

Since the AQF has three metadata bits, we can use a much simpler encoding.  
The AQF has three types of slots---remainders, extensions, and counters---and we need only to distinguish between the two types of ``extra'' slots that can follow a remainder: extension slots and counter slots.  In our encoding, both extension and counter slots have the \emph{is\_extension} bit set, and we use the \emph{is\_runend} bit to indicate whether the slot holds an extension or a counter.  This is safe because we indicate the end of a run by setting the \emph{is\_runend} bit on only the remainder slot of the last fingerprint in the run -- any other slots of that fingerprint may be free to use their own \emph{is\_runend} bit as needed to indicate if they are extension or counter slots.

\subsection{Dynamic \YesNo List Problem
}
\label{sec:dynamic-adaptiveqf}
We can extend \sysname to the dynamic \yesno list problems as follows.  First, we extend the filter to store an extra bit with each fingerprint, i.e. we extend each slot with one extra bit.  We then store all elements of both $Y$ and $N$ in the filter, performing adaptations to eliminate any fingerprint collisions that occur during insertion and using the extra bit to record which set each fingerprint belongs to.  We can now add new items to $Y$ and $N$ in the same way: add the item to the filter, performing any adaptation necessary to eliminate fingerprint collisions and tagging the items with their origin.  We can delete items by simply deleting them from the filter.  Deleting a fingerprint $f$ may mean that we can shorten other fingerprints in the filter that we extended because they collided with $f$.  Finding any such eligible fingerprints is easy and efficient because the \sysname stores fingerprints sorted lexicographically, so all the fingerprints that can be shortened will be stored in a contiguous run of slots containing $f$.

Like the quotient filter, the \sysname also supports growing and shrinking.  The data structure described in this section operates by initially provisioning the table with enough slots to hold a certain number of elements and adaptivity bits. In the case of the static \sysname filter described in \Cref{sec:blacklist-upper-bound}, when sizes $n$ and $m$ of the \yes list and \no list, respectively, are known ahead of time, we can predict the total space cost with high confidence, as per \Cref{thm:blacklist-failure-probability}. This tight allocation of the table is what allows us to match the lower bound from \Cref{thm:lower-bound}.

Alternatively, if the \yes and \no list sizes are not known in advance, we can instead fix two upper bounds $\hat{n}$ and $\hat{m}$ on their maximum allowed size, and construct the following \emph{dynamic} \yesno filter: we allocate our table as a function of $\hat{n}$ and $\hat{m}$ (instead of $n$ and $m$), and perform all the insertions into the \no list and queries of \yes list elements dynamically, as they are needed. The closer $n$ and $m$ get to $\hat{n}$ and $\hat{m}$, respectively, the closer the space cost will be to the space lower bound for a static filter on these \no list and \yes list sets. 

\subsection{Skewed and Adversarial Workloads}
\label{sec:skewed-workloads}

\newtext{
The basic \sysname structure is monotonically adaptive, i.e. it never repeats a false positive.  
The cost of never forgetting false positives is that, over time, the \sysname needs more and more slots to hold adaptivity information.  
Like the regular quotient filter, the cost of an insert into the \sysname is $\Theta(\log n / (1-\alpha)^2)$ w.h.p., where $n$ is the number of slots and $\alpha$ is the fraction of slots that are currently in use~\cite{BenderFaJo12}.
Thus, as the \sysname adapts, $\alpha$ approaches 1 and insertions performance can fall off a cliff.
In static and dynamic \yesno-list problems, this can be mitigated by making the filter large enough to accommodate the anticipated number of items.
}

\newtext{
However, for skewed and adversarial workloads, we can recover space used by adaptation, ensuring that the total space used by the filter remains constant over time.  This will compromise monotonicity, but will still ensure that the number of false positives from any sequence of $k$ queries is very close to $\epsilon k$ w.h.p.
}

\newtext{
The basic idea is to periodically rebuild the filter with a new hash function.  Rebuilding the filter puts the attacker back into the position of attacking a filter about which he has no information.  Thus we can drop any adaptivity information after the rebuild.  In other words, when we do the rebuild, each item will consume a single slot.
}

\newtext{
So, for example, we can build the filter, say, 10\% larger than necessary, run the filter until the extra space is consumed by adaptations, and then rebuild.  Furthermore, we can de-amortize the rebuild process.  See Bender et al.~ for details~\cite{BenderFaGo18}.
}

\section{Static \YesNo List Bounds}
\label{sec:theory}

In this section, we prove a space lower bound for solving the \textit{static} \yesno list problem and 
we show that we can use \sysname to build an optimal solution to the static \yesno list problem, up to low order terms.
Importantly, existing practical solutions to this problem (namely the Seesaw Counting Filter \cite{seesaw2022},
and the Bloomier Filter \cite{chazelle2004bloomier}) are either not always correct solutions to the problem, or their space cost is at least constant factor away from the lower bound.

Let $U$ be a finite universe of elements, and let $Y$ and $N$ be subsets of $U$, with $Y \cap N = \emptyset$. Let $\epsilon \in (0, 1)$. A \defn{\yesno filter} supports queries with the following guarantees: (i) every query for $y \in Y$ must answer \yes, (ii) every query for $z \in N$ must answer \no, and (iii) every query for $x \notin Y \cup N$ answers \yes with probability at most $\epsilon$. Notice that \yesno filters are static data structures.
Although it's possible to consider a dynamic version where the elements of either the \yes or \no set (or both) are inserted and deleted dynamically, in this section we do not study this scenario.

Throughout this section we will let $n = |Y|$ be the size of the \yes list, $m = |N|$ be the size of the \no list set, and $u = |U|$ be the size of the universe.

\subsection{Upper Bound for \YesNo Filters}
\label{sec:blacklist-upper-bound}

Next we give an upper bound for the static \yesno list problem, based on the \sysname. We maintain the notation of the previous section, and let $\mu := \epsilon m/n$ be the number of \no false positives per \yes element. Consider the following implementation, which we refer to as the \yesno \sysname.

\begin{enumerate}[leftmargin=*]
    \item Create an \sysname $F$ with capacity for $n$ elements, and $\calA(n, m, \epsilon)$ bits reserved for adaptivity, where
    \[
    \calA(n, m, \epsilon) := (1 + o(1)) n \log(1 + \mu) + O(n).
    \]
	\item Insert into $F$ every element from $Y$.
	\item Query $F$ on each element from $N$.
	\item If $F$ becomes full at any point before all queries are done, fail.
	\item Return $F$.
\end{enumerate}

Importantly, the queries in the final step fix all false positives from $N$. The resulting filter satisfies the requirements of a \yesno filter.

\begin{proposition}
The \yesno \sysname uses \\$(1 + o(1)) n \log(\max\{1/\epsilon, m/n\}) + O(n)$ bits of space.
\end{proposition}

\begin{proof}
The space cost is the sum of the space reserved for the remainders, plus the per-slot metadata bits, plus the space reserved for adaptivity bits. This is a total of $n \log(1/ \epsilon) + O(n) + \calA(n, m, \epsilon)$ bits. Observe that
\begin{align*}
\calA(n, m, \epsilon) &=
\left\lbrace
\begin{array}{ll}
O(n) & \text{if $\mu \leq 1$}\\
(1 + o(1))n(-\log(1/\epsilon) + \log(m/n)) + O(n) & \text{otherwise}
\end{array}
\right.
\end{align*}
Notice that $\mu \leq 1$ if and only if $1 / \epsilon \geq m/n$. Hence,
\begin{align*}
n \log(1/ \epsilon) + O(n) +& \calA(n, m, \epsilon)\\
&= \left\lbrace
\begin{array}{ll}
n \log(1/\epsilon) + O(n) & \text{if $1 / \epsilon \geq m/n$}\\
(1 + o(1))n \log(m/n) + O(n) & \text{otherwise}
\end{array}
\right.\\
&= (1 + o(1))n\log(\max\{1/\epsilon, m/n\}) + O(n).
\end{align*}
\end{proof}

Notice that the construction of the \yesno \sysname may fail if the space initially reserved is insufficient. The following theorem, the central result of this section, establishes that failure is unlikely.

\begin{theorem}
\label{thm:blacklist-failure-probability}
Suppose $\omega((\log^{3/2} n)/\sqrt{n}) \leq \mu \leq 2^{o(n)}$. Then, the number of adaptivity bits added to the \yesno \sysname is at most $\calA(n, m, \epsilon)$ with probability $1 - 1/\poly(n)$. In particular, the probability that the construction succeeds is $1 - 1 / \poly(n)$.
\end{theorem}

In the rest of this section, we prove this theorem. Let $A$ be the number of adaptivity bits needed after all the elements of $N$ are queried in step (5). We want to show that $A \leq \calA(n, m, \epsilon)$ with probability $1 - 1/\poly(n)$.

Let $h$ be fingerprint hash function. For any $x, y \in U$, let $\lcp(x, y)$ be the longest common prefix of $h(x)$ and $h(y)$. We decompose $A$ into two parts: Let $A_1$ and $A_2$ be the number of adaptivity bits added in steps 3 and 4, respectively. For any $y \in Y$, let $A_1(y)$ and $A_2(y)$ be the number of adaptivity bits added in step 3 and 4, respectively, to the fingerprint of $y$. Then, $A = A_1 + A_2$, and $A_i = \sum_{y \in Y} A_i(y)$.

\begin{lemma}
\label{lem:concentration-1}
We have $A_1 = O(n)$ with probability $1 - 1/\poly(n)$.
\end{lemma}
\begin{proof}
This follows directly from Lemma 9 in Bender et. al.~\cite{BenderFaGo18}.
\end{proof}

We will use the following basic result about the distribution of the maximum of a collection of independent geometric random variables.

\begin{lemma}[\cite{Eisenberg2008MaxGeometric}]
\label{lem:maximum-geometric}
Let $X_1, \dots, X_k$ be independent geometric random variables with parameter $1/2$. Let $M_k := \max_i X_i$. Then,
\[
0 \leq \expect{M_k} - \log(e) H_k \leq 1,
\]
where $H_k$ is the $k$-th harmonic number.
\end{lemma}



\begin{lemma}
\label{lem:expectation-2}
$\expect{A_2} \leq n \left(1 + \log(e) + \log\left(1 + \mu\right)\right)$
\end{lemma}

\begin{proof}
We say that $z \in N$ has a \defn{soft collision} with $y \in Y$ if the baseline fingerprints of $y$ and $z$ match, that is, $\lcp(y, z) \geq \log(n/\epsilon)$.



Fix an arbitrary $y \in Y$. We claim that $\expect{A_2(y)} \leq 1 + \log(e) + \log\left(1 + \mu\right)$. Let $C$ be the (random) numbers of $z \in N$ that have a soft collision with $y$. For each $z$ that has a soft collision with $y$, let $F(z) := \lcp(y, z) - \log(n/\epsilon)$; this is the smallest number of adaptivity bits that $y$ must have to fix a false positive $z$. Then,\\ $A_2(y) = \max\limits_{\text{$z$ has a soft collision with $y$}} F(z)$.

Observe that $F(z)$ is a geometric random variable with parameter $1/2$. Because all fingerprints are independent, the $F(z)$'s are independent. This is true even if $C$ is known. Thus, $A_2(y)$ conditioned on $C = k$ is identically distributed as the maximum of $k$ independent geometric random variables with parameter $1/2$. Call this maximum $M_k$. Then,
\[
\expect{A_2(y) \mid C = k} = \expect{M_k}.
\]
For $k = 0$, we have $\expect{A_2(y) \mid C = k} = 0$. For $k \geq 1$,
\begin{align*}
    \expect{A_2(y) \mid C = k} &\leq 1 + \log(e) H_k \tag{by \Cref{lem:maximum-geometric}}\\
    &\leq 1 + \log(e)(1 + \ln k) \tag{as $H_k \leq 1 + \ln k$ for $k \geq 1$}\\
    &= 1 + \log(e) + \log k.
\end{align*}
Hence, $\expect{A_2(y) \mid C} \leq 1 + \log(e) + \log(1 + C)$. Then,
\begin{align*}
    \expect{A_2(y)} &= \expect{\expect{A_2(y) \mid C}} \tag{by the tower rule}\\
    &\leq \expect{1 + \log(e) + \log(1 + C)}\\
    &\leq 1 + \log(e) + \log(1 + \expect{C}). \tag{by linearity of expectation and Jensen's inequality}
\end{align*}
Notice that $C = \sum_{z \in N} C(z)$, where $C(z)$ is an indicator random variable that is $1$ exactly when $z$ has a soft collision with $y$. Recall that $\expect{C(z)} = \epsilon / n$. By linearity of expectation, $\expect{C} = \epsilon m/n = \mu$, and, finally,
\[
\expect{A_2(y)} \leq 1 + \log(e) + \log\left(1 + \mu\right).
\]
\noindent
This concludes the proof of the claim. The lemma follows by summing over all $y \in Y$, and using linearity of expectation.
\end{proof}

To simplify notation, let $\mu = \epsilon m/n$ be the mean number of queries that have a soft collision with any fixed $x \in S$.

\begin{lemma}
\label{lem:concentration-2}
Suppose $\omega((\log^{3/2} n) / \sqrt{n}) \leq \mu \leq 2^{o(n)}$. Then, $A_2 \leq (1 + o(1)) \expect{A_2}$ with probability $1 - 1 / \poly(n)$.
\end{lemma}
\begin{proof}[Proof sketch]
Recall that $A_2 = \sum_{y \in Y} A_2(y)$. The proof is divided into two parts. First, we show that the random variables $A_2(y)$ are negatively associated (NA). Roughly speaking, this is because when some query from $N$ is a false positive due to fingerprint match with $y \in Y$, then the new adaptivity bits make following queries less likely to cause a false positive on $y$. Though intuitive, proving that the $A_2(y)$'s are NA is challenging because we can't directly apply standard theorems on negative association~\cite{Wajc2017NA}.

For each $y \in Y$, let $G(y)$ be the number of $z \in N$ such that $\lcp(y, z) \geq \log(n/\epsilon) + A_1(y)$, that is, elements that can cause a false positive due to a fingerprint match with $y$ (and only with $y$). The argument rests on the following four properties of the random variables involved:
\begin{enumerate}
    \item The $G(y)$'s are NA random variables.
    \item Conditioned on the $G(y)$'s, the $A_2(y)$'s are independent.
    \item The probability that $A_2(y)$ is large is a non-decreasing function of the $G(y)$. Formally, for every $\ell$, $\prob{A_2(y) \geq \ell \mid G(y) = k}$ is a non-decreasing function of $k$.
    \item Conditioned on $G(y)$, the random variable $A_2(y)$ is independent of the $G(y')$'s with $y' \neq y$.
\end{enumerate}

Once negative association of the $A_2(y)$'s is established, we are almost in the conditions of the Chernoff-Hoeffding inequality for the sum $\sum_{y \in Y} A_2(y)$. Unfortunately, there is one hypothesis that is not met, namely they $A_2(y)$'s are not deterministically bounded---in the worst case, an unbounded number of adaptivity bits may need to be added to some element $y$. This, however, is unlikely, as $A_2(y) = O(\log(\mu + n))$ with probability $1 - 1/\poly(n)$; this is by a Chernoff bound, and a tail bound on the geometric distribution. We can put this observation to work and circumvent the boundedness requirement of Chernoff-Hoeffding, using the following truncation trick: We define $A_2'(y) := \min\{A_2(y), O(\log(\mu + n))\}$, and apply the Chernoff-Hoeffding bound on the truncated sum $A_2' := \sum_{y \in Y} A_2'(y)$. Once concentration around the mean is established on $A_2'$, we conclude the proof by showing that, with high probability, no truncation is actually done, so the analysis on $A_2'$ applies to $A_2$ most of the time. Specifically:
\begin{enumerate}
    \item $A_2 = A_2'$ with probability $1 - 1/\poly(n)$;
    \item $\expect{A_2'} \leq \expect{A_2}$.
\end{enumerate}
\end{proof}

\begin{proof}[Proof of \Cref{thm:blacklist-failure-probability}]
The construction of the filter succeeds if and only if $A \leq \calA(n, m, \epsilon)$. Since $A = A_1 + A_2$, and by \Cref{lem:concentration-1}, \Cref{lem:expectation-2} and \Cref{lem:concentration-2} we have $A \leq O(n) + (1 + o(1))(n (1 + \log(e) + \log(1 + \mu))) = O(n) + (1 + o(1)) n \log(1 + \mu) = \calA(n, m, \epsilon)$, with probability $1 - 1/\poly(n)$.
\end{proof}

\subsection{Lower Bound for \YesNo Filters}

A lower bound for this problem was sketched out in Reviriego et al.~\cite{Reviriego2021BlacklistFilter}, but without a rigorous proof. Moreover, the lower bound as stated in that work is hard to compare against our upper bound using \sysname; here, we give an equivalent but more condensed lower bound.

\begin{theorem}
\label{thm:lower-bound}
Suppose $u \geq c(n^2/\epsilon + m^2)$, for some large enough constant $c > 0$, and $\epsilon \leq 1/2$. Then, the number of bits used by a static \yesno filter is at least
\[
n \log\left(\max\left\lbrace\frac{1}{\epsilon}, \frac{m}{n}\right\rbrace\right) + \log(e) \min\{\epsilon m, n\} + O(1).
\]
\end{theorem}

Before diving into the proof let us briefly discuss the lower bound. Dividing by $n$, we have that a \yesno filter uses at least
\[
\log\left(\max\left\lbrace 1/\epsilon, m/n\right\rbrace\right) + O(1) = \log(1/\epsilon) + \log(\max\{\epsilon m/n, 1\}) + O(1)
\]
bits per \yes element. For comparison, traditional filters have an information-theoretical lower bound of $\log(1/\epsilon) + O(1)$ bits per element. This can be interpreted as follows: When building a \yesno filter, we need to (i) record the $n$ \yes list elements while ensuring that at most an $\epsilon$ fraction of all other elements are incorrectly reported as present, and to (ii) record the $m$ \no elements. To accomplish (i) we need at least $\log(1/\epsilon)$ bits per element, just like a regular filter. The number of additional bits needed to accomplish (ii) depends on $\mu := \epsilon m/n$, which is the number of \no false positives per \yes element. When $\mu \leq 1$, only a small constant number of extra bits per \no element are needed; when $\mu > 1$, $\log(\mu)$ extra bits per \no element are needed. 


\begin{proof}[Proof sketch]
Let $\calI = \{(Y, N) \mid Y, N \subseteq U, Y \cap N = \emptyset, |Y| = n, |N| = m\}$ be the set of all inputs of the problem. Consider a static \yesno filter with false positive probability $\epsilon$, that uses at most $b$ bits on all inputs $(Y, N) \in \calI$. For any random bit string $r$, denote $\textsf{Inst}(Y, N, r)$ the instance of the data structure with randomness $r$ on input $(Y, N)$. We say that an instance $F$ \defn{represents} input $(Y, N)$ if there exists a string of random bits $r$ such that $F = \textsf{Inst}(Y, N, r)$.

The proof uses the following standard information-theoretic argument: Let $\calF$ be a subset of instances of the data structure, such that
\begin{enumerate}
    \item every input is represented by some instance of $\calF$, and
    \item for every $F \in \calF$, $F$ represents at most $c$ inputs.
\end{enumerate}
\noindent
When $\calF$ satisfies property (1), we say it \defn{covers} $\calI$; when it satisfies (2), we say it has \defn{covering parameter} $c$. Then, $|\calF| \geq |\calI| / c$. Because all instances from $\calF$ must be encoded by some state of the data structure, we have $2^b \geq |\calF|$, which in turn implies the space lower bound $b \geq \log |\calF| \geq \log(|\calI| / c)$. The challenge is finding a cover with small covering parameter $c$.

Fix any input $(Y_1, N_1)$. Let $\calF$ be the set of instances that err on at most a fraction $\epsilon$ of $\overline{U} := U \setminus (Y_1 \cup N_1)$. Let $F_r = \textsf{Inst}(Y_1, N_1, r)$. Because the \yesno filter errs with probability $\epsilon$, we have, for every $z \in \overline{U}$,
\[
\Pr_r\left[\text{$F_r$ errs on $z$}\right] \leq \epsilon.
\]
Summing over $z$, we find that
\[
\E_r\left[\text{number of $z \in \overline{U}$ on which $F_r$ errs}\right] \leq \epsilon |\overline{U}|.
\]
Hence, there exists some $r$ such that $F_r$ errs on less than a fraction $\epsilon$ of $\overline{U}$. Because this instance $F_r$ represents $(Y_1, N_1)$, this implies that there exists some instance from $\calF$ that represents $(Y_1, N_1)$, that is, $\calF$ covers $\calI$.

The number of inputs $(Y, N) \in \calI$ is
\[
|\calI| = {u \choose n}{u-n \choose m}.
\]
\noindent
We claim that $\calF$ has covering parameter
\begin{equation}
\label{eq:covering}
c \leq \min_{\epsilon' \in (0, \epsilon]} {n + \epsilon'(u - n - m) \choose n}{u - (n + \epsilon'(u - n - m)) \choose m}.
\end{equation}
\noindent
To see this, consider an $F \in \calF$ and let $(Y, N) \in \calI$ represented by $F$. Then, $c$ is upper bounded by the number of inputs $(Y', N')$ such that $Y' \subseteq \{x \in U \mid \text{$F$ answers \yes on $x$}\}$ and $N' \subseteq \{x \in U \mid \text{$F$ answers \no on $x$}\}$; these are all the inputs compatible with the answers of $F$. The first and second terms on the right-hand side of \Cref{eq:covering} are the number of ways to choose $Y'$ and $N'$, respectively, when $F$ errs on a fraction $\epsilon' \in (0, \epsilon]$ of elements from $\overline{U}$. Notice that the right-hand side of \Cref{eq:covering} is \emph{not} non-decreasing in $\epsilon'$, so the minimum is not necessarily attained at $\epsilon' = \epsilon$.

A lengthy calculation, which we omit here due to space constraints, shows that
\[
\log(|\calI|/c) \geq \min_{\epsilon' \in (0, \epsilon]} n\log(1/\epsilon') + \epsilon' m \log(e) + O(1).
\]
Finally, an analysis of the function $f(\epsilon') = n \log(1/\epsilon') + \epsilon' m \log(e) + O(1)$ with $\epsilon' \in (0, \epsilon]$ reveals that
\begin{itemize}
	\item if $\epsilon \leq n/m$, then $f$ attains its minimum at $\epsilon' = \epsilon$;
	\item if $\epsilon > n/m$, then $f$ attains its minimum at $\epsilon' = n/m$. 
\end{itemize}
Hence,
\begin{align*}
\log(|\calI|/c) \geq &\left\{
\begin{array}{cc}
n \log(1/\epsilon) + \epsilon m \log(e) + O(1) & \text{if $\epsilon \leq n/m$}\\
n \log(m/n) + n \log(e) + O(1) & \text{otherwise}
\end{array}
\right.\\
&= n \log(\max\{1/\epsilon, m/n\}) + \log(e) \min\{\epsilon m, n\} + O(1).
\end{align*}

\end{proof}


\section{Evaluation}\label{sec:evaluation}

The goal of the evaluation is to answer the following questions regarding the performance of the \sysname:
\begin{enumerate}[noitemsep, leftmargin=*]
    \item How does the insertion and query performance of the \sysname compare to other adaptive and non-adaptive filters?
    \item How well does the \sysname improve overall database performance, compared to other adaptive and non-adaptive filters?
    \item How much space does the \sysname use to adapt?
    \item How does the \sysname compare to prior solutions to the \yesno list problem?
    \item How does the false positive rate in the \sysname change during a dynamic workload?
    \item How fast can two \sysname instances be merged?
\end{enumerate}

\subsection{Results summary}
We compare the \sysname~\footnote{https://anonymous.4open.science/r/adaptiveqf-C438} against two state-of-the-art adaptive filters, the telescoping adaptive filter (TQF)~\cite{lee2021telescoping} and the adaptive cuckoo filter (ACF)~\cite{MitzenmacherPR20}.
We also include two non-adaptive filters, the quotient filter (QF)~\cite{PandeyBJP17} and the cuckoo filter (CF)~\cite{fan2014cuckoo}, as baselines to understand the overheads and benefits of adaptivity. The quotient and cuckoo filter are chosen as baselines as these are the filters upon which the adaptive filters used in our evaluation are developed.

We found that adaptivity is an extremely efficient way to reduce the false-positive rate of a filter. For example, on a Zipfian query workload, the \sysname is able to reduce the false-positive rate by about 100$\times$ for an additional cost of less than 1/1000th of a bit per item.  In contrast, a non-adaptive filter would need 7 bits per item to achieve the same false-positive rate reduction.

Absent any system, the \sysname has comparable space usage to the other filters.  For example, the \sysname uses more space than the cuckoo filter, but only by 1\%. 
Most notably, the \sysname performs at par with the quotient filter on which it is based, indicating little to no overhead for its adaptivity. On the other hand, the adaptive cuckoo filter and telescoping adaptive filter are significantly slower than their respective non-adaptive counterparts.

However, when the cost of a false positive is increased by including an on-disk database, the benefits of adaptivity become apparent. For example, when used to filter queries from a fixed dataset, the \sysname is able to learn the query set, seeing 10$\times$ fewer false positives over 200 million queries than the quotient filter and the cuckoo filter. This resulted in 4-7$\times$ faster queries. Furthermore, using the \sysname to filter queries in a B-tree databases had query throughput that was impervious to an adversarial query workload, whereas the throughput of non-adaptive filters dropped about 2$\times$ with the inclusion of an adversary representing a mere 1\% of queries.

Although the benefits of maintaining a low false positive rate during queries are shared by all three adaptive filters, the inclusion of a database reveals that the \sysname is much faster than other adaptive filters during insertions. The systems using the TQF and the ACF slow down significantly as the filters fill up due to frequent modifications of their on-disk backing stores. Between 85-90\% fullness, the \sysname averages 5$\times$ the insertion throughput of the ACF and 30$\times$ that of the TQF.

\para{Other results}
The \sysname matches the rate of change of false-positive rate during queries from a real-world dataset compared to other adaptive filters. The \sysname solves dynamic \yesno problems with changing sets. The \sysname supports fast merges and bulk insertions. Finally, despite being dynamic, the \sysname achieves similar or better space usage compared to purpose-build solutions for the static \yesno problem.

\subsection{Experimental setup}
One challenge we face is that the filters do not all support the same false-positive rates.  
%
Thus, we pick a target false-positive rate and configure each filter
to get as close as possible to the target false-positive rate without
sacrificing performance. This is in accordance with the prior research on evaluating filters~\cite{PandeyBJP17a,FanAnKa14,pandey2021vector}. Our target false-positive rate is $2^{-9}$ ($\approx 0.2\%$) which is the commonly used in most practical system configurations~\cite{PandeyBJP17,fan2014cuckoo}.

We configure the \qf-based filters (AQF, TQF, and QF) with 9-bit remainders. We use 12-bit fingerprints and blocks of size 4 in the cuckoo filter-based filters (ACF, CF). This results in all filters having a false-positive rate of $2^{-9}$. \newtext{We use MurmurHash2~\cite{Appleby} as the hash function for all filters. The \sysname does not require any special properties in the hash function compared to other filters.}


\para{Machine specification}
All experiments were run on an Intel(R) Xeon(R) Gold 6338 CPU @ 2.00GHz with 96 MiB L3 cache. The machine has 1 TB of memory, 64 CPUs, and 4 TB of SSD-based local instance storage, 64-bit platform. We restrict our runs to a single core. 

\subsection{Microbenchmarks}

We evaluate the performance of the filters in RAM. We create the filters with $2^{27}$ (134M) slots, which makes them substantially larger than the L3 cache on the machine where experiments are performed.
We fill each filter to 90\% load factor\footnote{Load factor is the ratio of the number of occupied slots over the total number of slots in the filter} and report the performance of the filter as a function of load factor.
Although all of the filters evaluated in our benchmarks support up to 95\% load factor, we restrict them to 90\% in order to give them room to store any additional data needed to adapt.

\para{Speed}
We evaluate adaptive filter performance on two fundamental operations: insertions and lookups. We evaluate insertions on uniformly random 64-bit keys, and lookups on both uniform-random and Zipfian distributions~\cite{corominasSo10}. In the Zipfian distribution, we use a Zipfian coefficient of 1.5 and a universe size of 10 million items. 
We do not count the time required to generate the input to the filters, only the time to insert and query items in the filters. This way we only measure the differences in filter performance.
We perform 200 million queries for both query distributions. The numbers reported for both insertions and queries are the average of 5 trials.

We perform the benchmarks in isolation of any overheads from the reverse maps. For the adaptive filters, which use reverse maps to obtain keys for adaptation, we pick valid arbitrary keys that will suffice in order to simulate having the reverse map present. We still measure hashing as part of the filters' performance because it is done independently of the reverse map and database.

\begin{figure}
   \captionsetup{skip=4pt}
   \begin{center}
        \ref*{insert-throughput-legend}
    \end{center}
    \begin{subfigure}{0.36\linewidth}
        \captionsetup{skip=0pt}
        \tikzsetnextfilename{insertion-throughput}
        \begin{tikzpicture}
            \begin{axis}[
                ThroughputBarGraph,
                ylabel = Ops/Sec,
                y label style = {at = {(axis description cs:0.45,0.5)}, anchor=south},
                width = 0.889\linewidth + 27pt,
            ]
            \addplot [AQFBarStyle, xshift=-0.3mm] table [ThroughputTable] {data/micro/aqf-0124-27-inserts.csv};
            \addplot [TQFBarStyle, xshift=-0.3mm] table [ThroughputTable] {data/micro/tqf-0124-27-inserts.csv};
            \addplot [ACFBarStyle, xshift=-0.3mm] table [ThroughputTable] {data/micro/acf-0124-27-inserts.csv};
            \addplot [QFBarStyle, xshift=0.3mm] table [ThroughputTable] {data/micro/qf-0124-27-inserts.csv};
            \addplot [CFBarStyle, xshift=0.3mm] table [ThroughputTable] {data/micro/cf-0124-27-inserts.csv};
         \end{axis}
         \draw (1.24,0) -- (1.24,2.45);
      \end{tikzpicture}
      \caption{Insertions}
      \label{fig:micro-insertion-throughput}
   \end{subfigure}%
   \begin{subfigure}{0.32\linewidth}
      \captionsetup{skip=0pt}
      \tikzsetnextfilename{uniform-query-throughput}
      \begin{tikzpicture}
         \begin{axis}[
               ThroughputBarGraph,
               ytick distance=20000000,
            ]
            \addplot [AQFBarStyle, xshift=-0.3mm] table [ThroughputTable] {data/micro/aqf-0124-27-queries.csv};
            \addplot [TQFBarStyle, xshift=-0.3mm] table [ThroughputTable] {data/micro/tqf-0124-27-queries.csv};
            \addplot [ACFBarStyle, xshift=-0.3mm] table [ThroughputTable] {data/micro/acf-0124-27-queries.csv};
            \addplot [QFBarStyle, xshift=0.3mm] table [ThroughputTable] {data/micro/qf-0124-27-queries.csv};
            \addplot [CFBarStyle, xshift=0.3mm] table [ThroughputTable] {data/micro/cf-0124-27-queries.csv};
         \end{axis}
         \draw (1.24,0) -- (1.24,2.45);
      \end{tikzpicture}
      \caption{Uniform Queries}
      \label{fig:micro-uniform-query}
   \end{subfigure}%
   \begin{subfigure}{0.32\linewidth}
      \captionsetup{skip=0pt}
      \tikzsetnextfilename{zipfian-query-throughput}
      \begin{tikzpicture}
         \begin{axis}[
               ThroughputBarGraph,
               ytick distance=75000000,
            ]
            \addplot [AQFBarStyle, xshift=-0.3mm] table [ThroughputTable] {data/micro/aqf-0124-27-zipfian.csv};
            \addplot [TQFBarStyle, xshift=-0.3mm] table [ThroughputTable] {data/micro/tqf-0124-27-zipfian.csv};
            \addplot [ACFBarStyle, xshift=-0.3mm] table [ThroughputTable] {data/micro/acf-0124-27-zipfian.csv};
            \addplot [QFBarStyle, xshift=0.3mm] table [ThroughputTable] {data/micro/qf-0124-27-zipfian.csv};
            \addplot [CFBarStyle, xshift=0.3mm] table [ThroughputTable] {data/micro/cf-0124-27-zipfian.csv};
         \end{axis}
         \draw (1.24,0) -- (1.24,2.45);
      \end{tikzpicture}
      \caption{Zipfian Queries}
      \label{fig:micro-zipfian-query}
   \end{subfigure}\\
   \caption{Micro operation throughput of filters absent any system. Done by simulating a sequence of operations first, then returning and performing the same operations but on the filter alone. Adaptive filters are compared on the left, while nonadaptive filters are shown as reference on the right.}
    \label{fig:micro-throughput}
\end{figure}
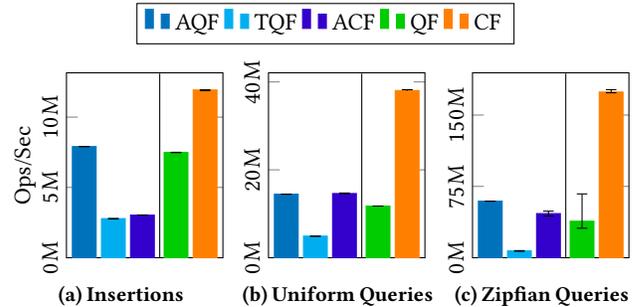

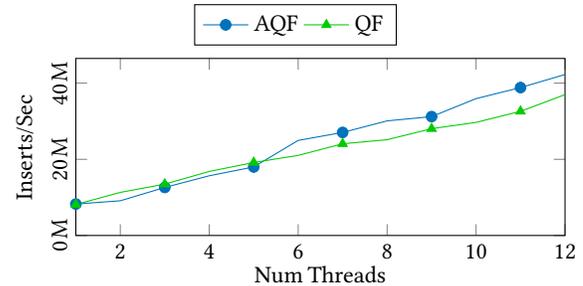
\begin{figure}
    \begin{center}
        \ref*{parallel-legend}
    \end{center}
    \begin{subfigure}{0.9\linewidth}
    \tikzsetnextfilename{parallel-inserts}
    \begin{tikzpicture}
        \begin{axis}[
            width = \linewidth + 13pt,
            height = 1.55in,
            ticks = major,
            try min ticks = 3,
            enlarge x limits = false,
            xtick scale label code/.code={},
            xticklabel = { \pgfmathprintnumber{\tick} },
            scaled y ticks = base 10:-6,
            ytick scale label code/.code={},
            yticklabel = { \pgfmathprintnumber{\tick}\hspace{0.08em}M },
            y tick label style = {rotate = 90},
            y label style = {at = {(axis description cs:0.12,0.5)}, anchor=south},
            x label style = {at = {(axis description cs:0.5,0.1)}, anchor=north},
            legend columns = -1,
            legend to name = {disk-insert-legend},
            xlabel = Num Threads,
            ylabel = Inserts/Sec,
            ymin = 0,
            legend to name = {parallel-legend},
        ]
        \addplot [AQFLineStyle] coordinates {(1, 8256526.292) (2, 9114274.957) (3, 12615243.48) (4, 15678175.46) (5, 17995763.7) (6, 24962751.35) (7, 27050802.12) (8, 30095969.5) (9, 31211133.86) (10, 35928421.52) (11, 38835647.69) (12, 42256115.62)};
        \addplot [QFLineStyle] coordinates {(1, 8137761.340027) (2, 11319296.261021) (3, 13482075.159105) (4, 16807081.851116) (5, 19167409.841285) (6, 21030817.726481) (7, 24078981.070894) (8, 25149149.613806) (9, 28037007.864366) (10, 29692500.601736) (11, 32600980.394062) (12, 36980176.247057)};
        \end{axis}
        \end{tikzpicture}
    \end{subfigure}%
    \vspace{-3mm}
    \caption{Parallel insertion throughput. $2^{26}$ slots in the filter.}
    \label{fig:parallel-inserts}
\end{figure}

\Cref{fig:micro-throughput} shows the throughput of adaptive and non-adaptive filters for insertions and queries. 
\newtext{The \sysname is based on the counting quotient filter. \sysname is not slower than the quotient filter, but it is slightly faster during both insertions and queries, indicating that the overhead of adaptivity in the \sysname is minimal.} Increased query speed may also be attributed in part to the slightly lower false positive rate resulting from adaptation. The CF has the highest insertion throughput among all the filters, which is consistent with previous research~\cite{PandeyBJP17a, FanAnKa14, pandey2021vector}. In exchange, the quotient filters offer fast resizing, mergeability, and efficient variable-length counters/values.

In contrast, there is a noticeable overhead in the ACF compared to the CF when it comes to queries due to the need to hash a given query multiple times. The CF can use a single hash function to obtain both the index and the tag of an item simultaneously. On the other hand, the tag of an item in the ACF depends on the location of the tag. This means that the ACF must first apply a hash to calculate the indexes of an item, and only then apply additional hashes to search for the tag.
Similarly, the TQF also sees a lower throughput due to the additional overhead of applying arithmetic coding to encode and decode hash selectors when making queries. Aside from queries, both the ACF and TQF have additional overhead during insertions for the same reason. As items naturally shift and move during insertions, the tags stored need to be rehashed in order to reflect their new locations.

When switching to Zipfian queries, all five filters benefit. Since the filters are easily large enough to overwhelm the machine's L3 cache, skewing the queries allows the cache to be more effective. The adaptive filters can also maintain a significantly lower false positive rate than the non-adaptive filters. However, the extra overhead in the ACF and TQF from the use of hash selectors puts a cap on how fast their queries can be, so the decreased false positive rate and increased cache friendliness have limited benefit. \newtext{The QF ends up having high variance in its zipfian query speed. This is a result of the high impact of locality in the QF, which uses linear probing in contrast to the CF. When popular items fall in small clusters, queries are fast, but if they are in larger clusters, query performance slows down. The \sysname does not see the same variance since it quickly adapts to any zipfian distribution regardless of its locality.}

\para{Space}
We evaluate the space efficiency of the filters by measuring the actual space needed to store items. We report the space efficiency at 90\% load factor. \newtext{This is space usage prior to any adaptation, so each filter contains the same number of items and uses the same number of slots.}
\Cref{tab:space-fpr} shows the empirical space usage and false-positive rate of different filters in these experiments. The space reported in the table is only the filter space. It does not include the space required by the reverse hash map. The \sysname has $\sim8-9\%$ space overhead compared to the non-adaptive quotient filter.

\para{Parallelism}
\newtext{The \sysname preserves thread safety from the counting quotient filter. It divides the slots in blocks of 4096 slots each and uses a lightweight spin lock for each block to avoid corruption.
During an insertion or an adaptation, each thread first acquires two locks on consecutive blocks, the block in which the item hashes and the next one. Two consecutive locks helps to avoid any corruption in case the shifting of items overflows into the next block.}

\newtext{It is also possible to execute mixed operations concurrently, but two modifications would be required. First, locks would also have to be acquired during queries, which would not be necessary if insertions and queries are performed in separate phases.  Second, if the database being used also supports concurrent inserts, the lock acquired during filter inserts would need to be held until the database insert is finished. This is to ensure the items in the same minirun are also inserted into the database in the same order as they are inserted into the filter. This is only necessary if there are mixed inserts and adaptations being done concurrently; in an insert-only workload, items in a minirun are identical in the filter until an adaptation happens, so the order of insertion into the database does not matter.}

\newtext{\Cref{fig:parallel-inserts} shows insertion throughput of the \sysname in isolation, as a function of the number of threads used, to demonstrate that the \sysname itself maintains good parallelism. We also show the performance of the QF for comparison. For this experiment, we use a filter of size $2^{26}$, and we configure the locks to span $2^{16}$ slots each. Therefore, there are $2^{10}$ locks and contention is low. We vary from 1 thread to 12 threads in the increments of 2. Both the \sysname and QF show almost linearly scaling with the increasing number of threads while the \sysname being slightly faster.}

\begin{table}[t]
\begin{tabular}{l l l }
\toprule
Filter & $-\log$(FPR) & Space (MB) \\
\midrule
Adaptive Quotient filter  & 9 & 203.610 \\
Telescoping Quotient filter  & 9 & 218.104 \\
Adaptive Cuckoo filter  & 9 & 201.402 \\
Quotient filter    & 9 & 186.818 \\
Cuckoo filter  & 9 & 201.401 \\
\bottomrule
\end{tabular}
\caption{\boldmath
  Empirical space usage and false-positive rate of filters used
  in the benchmarks. All filters were created with $2^{26}$ slots. 
  Space is given in MB. Negative $\log{FPR}$ means higher is better.}
\vspace{-1.5em}
  \label{tab:space-fpr}
\end{table}


\input{disk_throughput_plots}

\subsection{System benchmarks}



\begin{table}[t]
\centering
\begin{tabular}{l l l l l }
\hline
Filter & Size (log) & Map Inserts & Map Updates & Map Queries \\
\hline
AQF & 20 & 943718 & 0 & 0 \\
TQF & 20 & 943718 & 3608887 & 3356560 \\
ACF & 20 & 947815 & 584829 & 584829 \\
\hline
AQF & 24 & 15099494 & 0 & 0 \\
TQF & 24 & 15099494 & 56697889 & 52650676 \\
ACF & 24 & 15103591 & 9336669 & 9336669 \\
\hline
\end{tabular}
\caption{Number of reverse map accesses during insertions. The TQF and ACF make additional updates and queries in maintaining the reverse map. All filters of a given size were filled to 90\% load.}
\label{tab:db-access-during-inserts}
\end{table}

\begin{table}[t]
\centering
\begin{tabular}{l l l }
\hline
Reverse Map Setup & Inserts per Sec & Queries per Sec \\
\hline
Merged & $2.32 \times 10^5$ & $1.843 \times 10^7$ \\
Split & $1.12 \times 10^5$ & $1.819 \times 10^7$ \\
\hline
\end{tabular}
\caption{Comparison of merged vs split reverse map and database. Tests were run with 200M queries on a filter of size $2^{25}$. When the reverse map and database are split, inserts have to be done into each one independently, so insertion takes twice as long. However, due to the infrequency of false positives, querying only takes 1-2\% longer.}
\label{tab:reverse-map-setup-comparison}
\end{table}

\begin{table}[t]
\centering
\begin{tabular}{l l l }
\hline
Filter & CAIDA Queries/Sec & Shalla Queries/Sec \\
\hline
AQF & 31.6M & 15.8M \\
TQF & 8.4M & 6.7M \\
ACF & 34.1M & 18.3M \\
QF & 24.1M & 16.8M \\
CF & 119.2M & 62.2M \\
\hline
\end{tabular}
\caption{Query speed on real-world datasets after $2^{26}$ inserts.}
\label{tab:real-data-queries}
\vspace{-1.5em}
\end{table}

In this section, we evaluate the performance of \sysname as a front-end filter to a disk-based B-tree database. 
We create an instance of the disk-based B-tree by using the B-tree implementation from SplinterDB~\cite{ConwayGC20}. For these tests, we disable the $B^\varepsilon$-tree structure and its accompanying filters in SplinterDB, and use it as a filter-less on-disk dynamic $B$-tree. Because filters are frequently used alongside databases too large to fit in memory, storing data on disk using the B-tree is representative of real-world database systems.

We create an in-memory filter together with an on-disk map holding uniformly distributed keys with randomly generated values that represents a database. For the non-adaptive filters, the database holds the full set of keys-value pairs. For the adaptive filters, the database instead maps the fingerprints stored in the filter to their associated key-value pairs based on the optimization of using the reverse map as the database described in~\Cref{sec:implementation}.

To perform an insertion, a key is first inserted into the in-memory filter, and then the key-value pair is inserted into the database.
For non-adaptive filters, the key-value pair is inserted into the database directly. For adaptive filters, a fingerprint is obtained when inserting into the filter, and then the fingerprint-key-value triple is inserted into the database, mapping the fingerprint to the key-value pair.

To perform a query, the key is queried in the filter. If the filter returns ``negative,'' then that key is not in the database and no disk query is performed. If the filter returns a ``positive,'' the key and any corresponding data are retrieved from the database.
The non-adaptive filters do this by directly querying the database for the key, and a false positive occurs if the database was not able to find that key.
The adaptive filters instead query the database for the fingerprint found in the filter, and a false positive occurs if the key stored in the database does not match the key that was queried for. The adaptive filters can then use the returned key to adapt the filter so that the queried key no longer returns ``positive.''

\para{Insertion performance}
\Cref{fig:disk-inserts} shows insertion throughput of the database as a function of the filter load factor. For this experiment, we create filters with $2^{25}$ slots and insert keys from a uniform random distribution until the filters are 90\% full.
At 1\% progress intervals, we record the amount of time taken and calculate the insertion throughput over that interval. 

The system has similar performance when using the \sysname compared to the non-adaptive QF and CF filters. This shows that there is little to no overhead of using the adaptive filter on the insertion performance of the system.

Since insertions into the B-tree are the main bottleneck, all 5 filters start with roughly equal insertion throughput. However, the ACF and TQF fall off over time. This is due to the cost of maintaining the reverse map. \Cref{tab:db-access-during-inserts} shows the number of additional accesses to the reverse map done by the adaptive filters. As fingerprints are inserted into the \sysname, no entries for previous insertions need to be modified in the reverse map.
As the ACF fills up, it needs to do a large number of kick outs. Since the tag being stored to represent an item changes depending on its location in the filter, it is not sufficient to simply move a tag when performing a kick out. Instead, every kick out requires an expensive query to the backing map so that a new fingerprint can be hashed. The frequency of kick outs increases with load factor.
When a fingerprint is inserted into the TQF, it may cause other fingerprints to shift. The reverse map implemented with the TQF is based on location -- keys are stored alongside their fingerprints and thus need to shift with them. The constant shifting of fingerprints induces many additional reverse map accesses.

\para{Adversarial query performance}
In ~\Cref{fig:cache-adv-queries}, we measure the effect of a query-only adversary on system throughput. Even if the overall query distribution is uniform, an attacker can artificially skew the distribution by skewing their own queries. An adversary can detect the latency difference between negative and positive queries (including false positives), and even without knowledge of the actual insertion set, record a list of positive queries. They can then repeat these queries to intentionally induce I/Os. 
Even in a system with a cache, the adversary needs only collect enough false positives to overload the cache, then proceed to cycle between these queries to render the cache ineffective.

\newtext{In this experiment, we perform 200M queries. The first 100M give the adversary time to collect false positives. We then measure the average query throughput over the next 100M queries. We vary both the cache size and the frequency of adversarial queries over different trials. \Cref{fig:cache-adv-queries} has five plots, each plot showing the results of the experiment on a different cache size ranging from 1.5\%, 3\%, 6\%, 12\%, to 25\% the size of the input dataset. There are marginal improvements to the performance of the non-adaptive filters when provided a larger cache. However, a small fraction of adversarial queries has a disproportionately large impact on system performance even in the experiment with the largest cache. With the cache holding 25\% of the dataset, an adversary representing less than 1.5\% of the total queries can cause query throughput for the entire system to drop by 2$\times$ that of normal operation, which is already lower than that of the \sysname. This increases to 3$\times$ with 3\% adversarial influence and up to 10$\times$ with 10\% adversarial influence.}

The \sysname offers high and consistent query performance irrespective of the frequency of adversarial queries. Even without adversarial queries the AQF has comparable query performance to the non-adaptive filters. But in the presence of adversarial queries it can offer up to an order of magnitude higher query performance.

\para{Merged vs. split reverse map}
As discussed in ~\Cref{sec:implementation}, the reverse map and database can be merged into a single data structure so that reverse map inserts and queries do not incur additional overhead. This makes the database unable to perform range queries. The split reverse map and database setup, however, does support range queries. ~\Cref{tab:reverse-map-setup-comparison} compares the two setups to show the overhead of using the split setup in the case that one would like to use range queries. The insertion throughput is halved due to needing to insert into both the database and reverse map individually. However, query throughput is affected by only about 1\% due to the infrequency of false positives on general workloads.

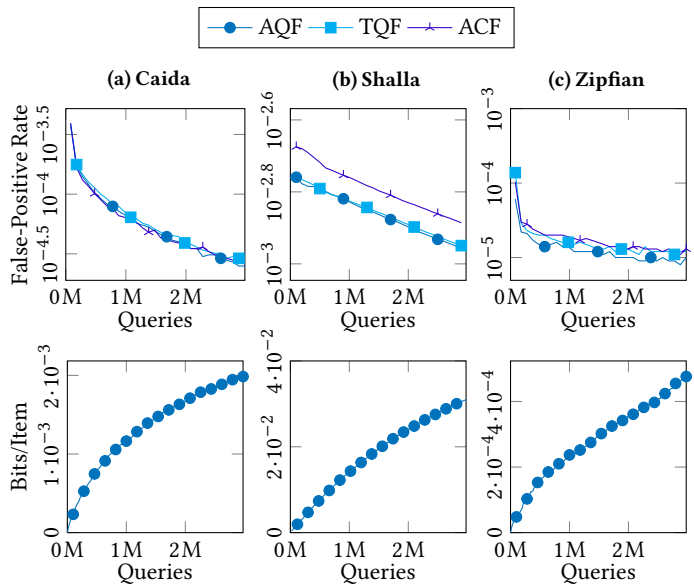
\begin{figure}
   \captionsetup{skip=4pt}
   \begin{center}
      \ref*{workload-legend}
   \end{center}
   \vspace{-1mm}
   \hspace*{-5mm}
   \begin{subfigure}{0.4\linewidth}
      \captionsetup{skip=0mm, margin={-5mm, 0mm}}
      \caption{Caida}
      \tikzsetnextfilename{workload-caida-fp}
      \begin{tikzpicture}
         \begin{axis}[
               WorkloadGraph,
               ylabel = False-Positive Rate,
               y label style={at={(axis description cs:0.35,.5)},anchor=south},
               ymode = log,
               legend to name={workload-legend},
               width = 0.889\linewidth + 27pt,
            ]
            \addplot [AQFLineStyle, mark repeat = 9, mark phase = 8] table [WorkloadTableFP] {data/workload/caida/aqf-averaged.csv};
            \addplot [TQFLineStyle, mark repeat = 9, mark phase = 2] table [WorkloadTableFP] {data/workload/caida/tqf-averaged.csv};
            \addplot [ACFLineStyle, mark repeat = 9, mark phase = 5] table [WorkloadTableFP] {data/workload/caida/acf-averaged.csv};
         \end{axis}
      \end{tikzpicture}
      \label{fig:workload-caida-fp}
   \end{subfigure}%
   \hspace*{-1mm}
   \begin{subfigure}{0.4\linewidth}
      \captionsetup{skip=0mm, margin={1mm, 0mm}}
      \vspace{2mm}
      \caption{Shalla}
      \vspace{-2mm}
      \tikzsetnextfilename{workload-shalla-fp}
      \begin{tikzpicture}
         \begin{axis}[
               WorkloadGraph,
               ylabel =,
               ymode = log,
               ymin = 0.0009,
               ymax = 0.0027,
            ]
            \addplot [AQFLineStyle, mark repeat = 8, mark phase = 1] table [WorkloadTableFP] {data/workload/shalla/aqf-averaged.csv};
            \addplot [TQFLineStyle, mark repeat = 8, mark phase = 5] table [WorkloadTableFP] {data/workload/shalla/tqf-averaged.csv};
            \addplot [ACFLineStyle, mark repeat = 8, mark phase = 0] table [WorkloadTableFP] {data/workload/shalla/acf-averaged.csv};
         \end{axis}
      \end{tikzpicture}
      \label{fig:workload-shalla-fp}
   \end{subfigure}%
   \hspace*{-5mm}
   \begin{subfigure}{0.4\linewidth}
      \captionsetup{skip=0mm, margin={1mm, 0mm}}
      \vspace{2.5mm}
      \caption{Zipfian}
      \vspace{-2.5mm}
      \tikzsetnextfilename{workload-zipfian-fp}
      \begin{tikzpicture}
         \begin{axis}[
               WorkloadGraph,
               ylabel =,
               ymode = log,
               ymax = 0.001,
            ]
            \addplot [AQFLineStyle, mark repeat = 9, mark phase = 6] table [WorkloadTableFP] {data/workload/zipfian/aqf-averaged.csv};
            \addplot [TQFLineStyle, mark repeat = 9, mark phase = 0] table [WorkloadTableFP] {data/workload/zipfian/tqf-averaged.csv};
            \addplot [ACFLineStyle, mark repeat = 9, mark phase = 3] table [WorkloadTableFP] {data/workload/zipfian/acf-averaged.csv};
         \end{axis}
      \end{tikzpicture}
      \label{fig:workload-zipfian-fp}
   \end{subfigure}\\
   \vspace*{-3mm}
   \hspace*{-5mm}
   \begin{subfigure}{0.4\linewidth}
      \captionsetup{skip=0pt}
      \begin{tikzpicture}
         \begin{axis}[
               WorkloadGraph,
               ylabel = Bits/Item,
               y label style={at={(axis description cs:0.35,.5)},anchor=south},
               ymin = 0,
            ]
             \addplot [AQFLineStyle, mark repeat = 9, mark phase = 6] table [WorkloadTableSpace] {data/workload/caida/aqf.csv};
         \end{axis}
      \end{tikzpicture}
      \label{fig:workload-caida-space}
   \end{subfigure}%
   \hspace*{-1mm}
   \begin{subfigure}{0.4\linewidth}
      \captionsetup{skip=0pt}
      \begin{tikzpicture}
         \begin{axis}[
               WorkloadGraph,
               ylabel =,
               ymin = 0,
               ymax = 0.04,
            ]
             \addplot [AQFLineStyle, mark repeat = 9, mark phase = 6] table [WorkloadTableSpace] {data/workload/shalla/aqf.csv};
         \end{axis}
      \end{tikzpicture}
      \label{fig:workload-shalla-space}
   \end{subfigure}%
   \hspace*{-5mm}
   \begin{subfigure}{0.4\linewidth}
      \captionsetup{skip=0pt}
      \begin{tikzpicture}
         \begin{axis}[
               WorkloadGraph,
               ylabel =,
               ymin = 0,
            ]
             \addplot [AQFLineStyle, mark repeat = 9, mark phase = 6] table [WorkloadTableSpace] {data/workload/zipfian/aqf.csv};
         \end{axis}
      \end{tikzpicture}
      \label{fig:workload-zipfian-space}
   \end{subfigure}
   \caption{False-positive rate (FPR) and additional space usage over time. FPR is on log scale.  The adaptive filters are able to almost immediately reduce their FPR by up to 100$\times$ on a skewed workload. The \sysname is able to achieve a marginally lower or equivalent FPR at negligible space overhead.}
    \label{fig:app-workloads}
\end{figure}
\subsection{Adaptivity rate for real-world datasets}

In application benchmarks, we use real-world datasets to evaluate the rate of change of false positive rate and space usage in adaptive filters in the presence of queries.

To evaluate the false positive rate over time, we first construct all three adaptive filters and fill them to 90\% load factor. We then construct a query set that will be performed over time, and the filters will adapt to the false positive queries. We also construct multiple independent query sets from the Zipfian distribution, which we use to compute the instantaneous false positive rate. The filters do not adapt while measuring the instantaneous false positive rate.

We perform a total of 3 million queries when the filters adapt to measure the rate of change of false positive rate and the space usage. We compute the instantaneous false positive rate and space usage after every 1\% of queries. 
To compute the instantaneous false positive rate, we construct 100 independent query sets from a Zipfian distribution. During the false positive computation, we turn off the adaption in the filters. Therefore, filters only adapt during normal queries and  do not adapt while computing the false positive rate.

We use three different datasets for application workloads. The first dataset is synthetic and generated from the Zipfian distribution (with Zipfian constant 1.5) on a universe of size 1 billion.
The second dataset is CAIDA passive traces~\cite{caida2016}, a set of anonymized network traces collected by the Center for Applied Internet Data Analysis between April 2008 and January 2019. 
The third database, the Shalla~\cite{shalla2021} block list, is a list of about 3 million malicious URLs compiled by Shalla Secure Services. For our experiments, we perform insertions and queries from the Shalla list.

\Cref{fig:app-workloads}  shows the rate of change of false positive rate and space during queries in adaptive filters.
The false positive rate immediately drop for all three filters. This is because the filters adapt to hot items early in the query sequence. 
Later on, the filters adapt to infrequent items, each of which brings a smaller drop in the false positive rate.

The drop in false positive rate over time is similar for all three adaptive filters for the Caida and Shalla datasets. These two datasets are not very skewed, and therefore the strong adaptivity advantage of \sysname over the TQF and ACF is not very apparent. 

On Zipfian queries, the false positive rate for all three filters drops equally. However, over time the false positive rate in the \sysname drops to lower than the TQF and ACF. This shows that the strong adaptivity guarantees in the \sysname lead to a lower false-positive rate over time. 
The TQF and ACF do not adapt completely the first time they encounter the false positive, which can result in subsequent false positive results when colliding with the already adapted key.

The space increase (in bits/item) increases similarly with all three adaptive filters. However, the \sysname has a lower initial space usage compared to the TQF and ACF (refer ~\Cref{tab:space-fpr}), and therefore the overall space usage of the \sysname is lower.

For actual query throughput, we list the numbers in \Cref{tab:real-data-queries}. This includes costs incurred by occasional queries to the database. The \sysname has comparable query throughput to both the AQF and QF. All filters benefit from cache-friendliness induced by the skewed distribution of CAIDA's queries. However, the \sysname and ACF see more improvement than the QF due to their ability to adapt to the most popular false positive queries in the distribution.

\subsection{Dynamic workloads}

Apart from static workloads, we also evaluate the change of false positive rate in the \sysname in the presence of deletions, insertions, and queries over time. This simulates the real-world use cases where the items in the \yes list change over time. We do not include other adaptive filter implementations, TQF and ACF, in this experiment as they do not support deletes.


Like the application workload, we perform 3 million queries and compute instantaneous false-positive rates after every 1\% of queries. At 10\% intervals, we delete and replace 20\% of the items. To compute the instantaneous false-positive rate, We use 1 million queries from the same Zipfian distribution without adapting.

\Cref{fig:workload-churn-fp} shows the false positive rate over time in the presence of deletions, insertions, and queries. Every 10\% of the operations, we introduce a massive churn in which 20\% of the items in the filter are replaced.
There are a couple of spikes in the false positive rate that coincide with the churns. They are caused when one of the inserted items causes a popular query to become a false positive. But the filter quickly adapts to the new item, and the false positive rate once again drops very low.
TQF and ACF are excluded from these experiments as those implementations do not support deletions.

In these experiments, we lose strong adaptivity. This is a deliberate choice and not a limitation of the \sysname. We can support strong adaptivity in the presence of updates to the \no list and \yes list by associating a small value to the fingerprints as described in~\Cref{sec:dynamic-adaptiveqf}. Strong adaptivity can be preserved in the presence of deletions by setting an item's counter to zero instead of deleting the item completely. We chose in these experiments not to preserve strong adaptivity in order to demonstrate the filter's ability to maintain a low false positive rate in dynamic environments and because the highly dynamic nature of these experiments would make the extra space usage impractical.

\begin{figure}
\captionsetup{skip = 4pt}
\begin{minipage}[t]{.9\linewidth}
   \begin{tikzpicture}
      \begin{axis}[
            ChurnGraph,
            ylabel = False-Positive Rate,
            xlabel = Operations,
            y label style={at={(axis description cs:0.1,.5)},anchor=south},
         ]
          \addplot [AQFLineStyle,
             mark repeat = 15,
             mark phase = 14,
             ] table [WorkloadTableFP] {data/workload/churn/aqf4.csv};
      \end{axis}
   \end{tikzpicture}
   \vspace{-5mm}
   \caption{False-positive rate over time on dynamic workloads. Churning points are marked.}
   \label{fig:workload-churn-fp}
\end{minipage}
   \begin{center}
      \ref*{blacklist-legend}
   \end{center}
\begin{minipage}[t]{0.9\linewidth}
   \begin{tikzpicture}
      \begin{axis}[
            BlacklistGraph,
            ylabel = Space (bytes),
            xmode = log,
            log basis x = 10,
            ymode = log,
            ymin = 1,
            y label style={at={(axis description cs:0.1,.5)},anchor=south},
         ]
         \addplot [AQFLineStyle] table [BlacklistTableSpace] {data/workload/blacklist/aqf.csv};
         \addplot [CBFLineStyle] table [BlacklistTableSpace] {data/workload/blacklist/cbf.csv};
      \end{axis}
   \end{tikzpicture}
   \vspace{-4mm}
   \caption{Space usage with varying ratio of \no list to \yes list. Space is on a log scale.}
   \label{fig:blacklist-workload}
\end{minipage}%
\vspace{-1em}
\end{figure}

\subsection{Merge and bulk load performance}
Filters are often used to build inverted text indexes on genomics data~\cite{PandeyABFJP18Cell}, where they are merged with other filters during compactions. Therefore, mergeability is a critical feature in filters for easy adoption in database systems. 

The QF supports efficient merging, and the \sysname supports efficient merging by extension, since we do not store any auxiliary hash encoding information. In contrast, merging in the TQF and ACF is not straightforward due to the hash selectors obscuring the original keys. To evaluate the merge performance of the \sysname, we use an in-memory hash table as the reverse map because we just want to evaluate the filter's merging speed. Note that merging two reverse maps is easy, because minirun lists sharing an ID can be concatenated, so long as miniruns in the filter are also concatenated during merging.

\begin{table}[t]
\centering
\begin{tabular}{l l }
\hline
Operation & Time per item ($\mu s$) \\
\hline
Insert into filter                & 0.520092 \\
Insert into half-size filter      & 0.353332 \\
Merge two half-size filters       & 0.039147 \\
Sort in hash order (qsort)        & 0.348060 \\
Bulk insert                       & 0.019569 \\
\hline
\end{tabular}
\caption{\boldmath Average latency for inserting items into an \sysname with $2^{26}$ slots until 90\% full, inserting into two \sysname instances with $2^{25}$ slots each and then merging, and sorting items beforehand and then bulk inserting. }
\label{tab:merge}
\vspace{-1.5em}
\end{table}

We also evaluate bulk loading in the \sysname, where the entire list of items is known. We find that the raw execution time of merging and bulk inserting is extremely low. For bulk loading, we would prepare by first sorting the items in hash order.

\Cref{tab:merge} shows the merge and bulk-build performance of the \sysname.  Inserting into smaller filters and merging is about $25\%$ faster than directly inserting into a full size filter.  Sorting followed by bulk building is about $10\%$ faster than merging.  In this particular test we used the C library function qsort.  More specialized sorting functions for a given situation may be even faster.  Merging is slower than bulk loading due to needing to compare quotient-remainder pairs between the two merged filters. There is also an overhead in identifying runs when stepping through the filter.

\subsection{Space comparison to static \YesNo solution}
\Cref{fig:blacklist-workload} shows the space usage of CRLlite~\cite{Larisch2017}, a custom-built and static \yesno list solution based on the cascading Bloom filter, and the \sysname.
The space usage of the \sysname while being dynamic is always smaller or similar to CRLite.
For the evaluation, we fix the aggregate size of the \no list and \yes list to 1 million items and evaluate the space with changing ratio of the \no list and \yes list.

\subsection{Non-adaptive filter additional space}
The adaptive filters have higher space usage (due to the overhead of adaptivity) compared to the non-adaptive filters. Therefore, we performed an experiment where we configured the QF and CF with a higher number of bits to give them extra space and lower false-positive rate.
With extra space, the uniform query performance of the CF increases by $1$\%, and the Zipfian query performance increases by $0.3$\%. Similar performance gains are seen for the QF. Therefore, even with extra space and a lower false-positive rate, the CF-based system is $20$\% slower compared to the \sysname-based system.



\section{Conclusion}\label{sec:conclusion}

We introduce \sysname in this paper. The \sysname is the first strongly adaptive filter which support high throughput operations using single-hashing and quotienting.  Using the adaptive filters in the system we can increase the overall system throughput by avoiding repeated unnecessary accesses to the backing stores (or other slower storage). The strongly adaptive filters guarantees consistently low false positives rate over time on dynamic workloads.

Traditional filters have been the go to data structure for over five decades. However, traditional filters lose their benefits in the presence of modern skewed and adverserial real-world workloads. Today's applications need practical adaptive filters that can offer strong theoretical guarantees and high performance independent of the data distribution to quickly and efficiently perform complex analyses on large-scale data.




\section*{Acknowledgments}
This research is funded in part by
NSF grant OAC 2339521.


\newpage

\bibliographystyle{ACM-Reference-Format}

\bibliography{bibliography}


\begin{thebibliography}{94}


\ifx \showCODEN    \undefined \def \showCODEN     #1{\unskip}     \fi
\ifx \showDOI      \undefined \def \showDOI       #1{#1}\fi
\ifx \showISBNx    \undefined \def \showISBNx     #1{\unskip}     \fi
\ifx \showISBNxiii \undefined \def \showISBNxiii  #1{\unskip}     \fi
\ifx \showISSN     \undefined \def \showISSN      #1{\unskip}     \fi
\ifx \showLCCN     \undefined \def \showLCCN      #1{\unskip}     \fi
\ifx \shownote     \undefined \def \shownote      #1{#1}          \fi
\ifx \showarticletitle \undefined \def \showarticletitle #1{#1}   \fi
\ifx \showURL      \undefined \def \showURL       {\relax}        \fi
\providecommand\bibfield[2]{#2}
\providecommand\bibinfo[2]{#2}
\providecommand\natexlab[1]{#1}
\providecommand\showeprint[2][]{arXiv:#2}

\bibitem[Almeida et~al\mbox{.}(2007)]%
        {AlmeidaBaPr07}
\bibfield{author}{\bibinfo{person}{Paulo~S{\'e}rgio Almeida},
  \bibinfo{person}{Carlos Baquero}, \bibinfo{person}{Nuno Pregui{\c{c}}a},
  {and} \bibinfo{person}{David Hutchison}.} \bibinfo{year}{2007}\natexlab{}.
\newblock \showarticletitle{Scalable {B}loom filters}.
\newblock \bibinfo{journal}{\emph{Journal of Information Processing Letters}}
  \bibinfo{volume}{101}, \bibinfo{number}{6} (\bibinfo{year}{2007}),
  \bibinfo{pages}{255--261}.
\newblock


\bibitem[Alsubaiee et~al\mbox{.}(2014)]%
        {AlsubaieeBeBo14}
\bibfield{author}{\bibinfo{person}{Sattam Alsubaiee},
  \bibinfo{person}{Alexander Behm}, \bibinfo{person}{Vinayak Borkar},
  \bibinfo{person}{Zachary Heilbron}, \bibinfo{person}{Young-Seok Kim},
  \bibinfo{person}{Michael~J Carey}, \bibinfo{person}{Markus Dreseler}, {and}
  \bibinfo{person}{Chen Li}.} \bibinfo{year}{2014}\natexlab{}.
\newblock \showarticletitle{Storage management in {AsterixDB}}.
\newblock \bibinfo{journal}{\emph{Proceedings of the VLDB Endowment}}
  \bibinfo{volume}{7}, \bibinfo{number}{10} (\bibinfo{year}{2014}),
  \bibinfo{pages}{841--852}.
\newblock


\bibitem[{Apache}({[n.\,d.]})]%
        {Cassandra}
\bibfield{author}{\bibinfo{person}{{Apache}}.}
  \bibinfo{year}{[n.\,d.]}\natexlab{}.
\newblock \bibinfo{title}{{Cassandra}}.
\newblock \bibinfo{howpublished}{\url{http://cassandra.apache.org}}.
\newblock


\bibitem[Appleby(2016)]%
        {Appleby}
\bibfield{author}{\bibinfo{person}{Austin Appleby}.}
  \bibinfo{year}{2016}\natexlab{}.
\newblock \bibinfo{title}{SMHasher source code in {C}++}.
\newblock
\newblock
\urldef\tempurl%
\url{https://github.com/aappleby/smhasher}
\showURL{%
\tempurl}


\bibitem[Beame et~al\mbox{.}(2014)]%
        {BeameKo14}
\bibfield{author}{\bibinfo{person}{Paul Beame}, \bibinfo{person}{Paraschos
  Koutris}, {and} \bibinfo{person}{Dan Suciu}.}
  \bibinfo{year}{2014}\natexlab{}.
\newblock \showarticletitle{Skew in Parallel Query Processing}. In
  \bibinfo{booktitle}{\emph{Proceedings of the 33rd ACM SIGMOD-SIGACT-SIGART
  Symposium on Principles of Database Systems}} (Snowbird, Utah, USA)
  \emph{(\bibinfo{series}{PODS '14})}. \bibinfo{publisher}{Association for
  Computing Machinery}, \bibinfo{address}{New York, NY, USA},
  \bibinfo{pages}{212–223}.
\newblock
\showISBNx{9781450323758}
\urldef\tempurl%
\url{https://doi.org/10.1145/2594538.2594558}
\showDOI{\tempurl}


\bibitem[Bender et~al\mbox{.}(2021)]%
        {BenderDaFa21}
\bibfield{author}{\bibinfo{person}{Michael~A. Bender}, \bibinfo{person}{Rathish
  Das}, \bibinfo{person}{Mart\'{\i}n Farach-Colton}, \bibinfo{person}{Tianchi
  Mo}, \bibinfo{person}{David Tench}, {and} \bibinfo{person}{Yung~Ping Wang}.}
  \bibinfo{year}{2021}\natexlab{}.
\newblock \showarticletitle{Mitigating False Positives in Filters: to Adapt or
  to Cache?}. In \bibinfo{booktitle}{\emph{Proc. 2nd Symposium on Algorithmic
  Principles of Computer System (APoCS)}}.
\newblock


\bibitem[Bender et~al\mbox{.}(2018)]%
        {BenderFaGo18}
\bibfield{author}{\bibinfo{person}{Michael~A. Bender}, \bibinfo{person}{Martin
  Farach{-}Colton}, \bibinfo{person}{Mayank Goswami}, \bibinfo{person}{Rob
  Johnson}, \bibinfo{person}{Samuel McCauley}, {and} \bibinfo{person}{Shikha
  Singh}.} \bibinfo{year}{2018}\natexlab{}.
\newblock \showarticletitle{Bloom Filters, Adaptivity, and the Dictionary
  Problem}. In \bibinfo{booktitle}{\emph{Proc.\ 59th Annual IEEE Symposium on
  Foundations of Computer Science (FOCS)}}. \bibinfo{address}{Paris, France},
  \bibinfo{pages}{182--193}.
\newblock


\bibitem[Bender et~al\mbox{.}(2012a)]%
        {BenderFaJo12a}
\bibfield{author}{\bibinfo{person}{Michael~A. Bender}, \bibinfo{person}{Martin
  Farach-Colton}, \bibinfo{person}{Rob Johnson}, \bibinfo{person}{Russell
  Kaner}, \bibinfo{person}{Bradley~C. Kuszmaul}, \bibinfo{person}{Dzejla
  Medjedovic}, \bibinfo{person}{Pablo Montes}, \bibinfo{person}{Pradeep
  Shetty}, \bibinfo{person}{Richard~P. Spillane}, {and} \bibinfo{person}{Erez
  Zadok}.} \bibinfo{year}{2012}\natexlab{a}.
\newblock \showarticletitle{Don't Thrash: How to Cache Your Hash on Flash}.
\newblock \bibinfo{journal}{\emph{Proceedings of the VLDB Endowment}}
  \bibinfo{volume}{5}, \bibinfo{number}{11} (\bibinfo{year}{2012}).
\newblock


\bibitem[Bender et~al\mbox{.}(2012b)]%
        {bender2012don}
\bibfield{author}{\bibinfo{person}{Michael~A Bender}, \bibinfo{person}{Martin
  Farach-Colton}, \bibinfo{person}{Rob Johnson}, \bibinfo{person}{Russell
  Kraner}, \bibinfo{person}{Bradley~C Kuszmaul}, \bibinfo{person}{Dzejla
  Medjedovic}, \bibinfo{person}{Pablo Montes}, \bibinfo{person}{Pradeep
  Shetty}, \bibinfo{person}{Richard~P Spillane}, {and} \bibinfo{person}{Erez
  Zadok}.} \bibinfo{year}{2012}\natexlab{b}.
\newblock \showarticletitle{Don't thrash: how to cache your hash on flash}.
\newblock \bibinfo{journal}{\emph{Proceedings of the VLDB Endowment}}
  \bibinfo{volume}{5}, \bibinfo{number}{11} (\bibinfo{year}{2012}),
  \bibinfo{pages}{1627--1637}.
\newblock


\bibitem[Bender et~al\mbox{.}(2012c)]%
        {BenderFaJo12}
\bibfield{author}{\bibinfo{person}{Michael~A. Bender}, \bibinfo{person}{Martin
  Farach-Colton}, \bibinfo{person}{Rob Johnson}, \bibinfo{person}{Russell
  Kraner}, \bibinfo{person}{Bradley~C. Kuszmaul}, \bibinfo{person}{Dzejla
  Medjedovic}, \bibinfo{person}{Pablo Montes}, \bibinfo{person}{Pradeep
  Shetty}, \bibinfo{person}{Richard~P. Spillane}, {and} \bibinfo{person}{Erez
  Zadok}.} \bibinfo{year}{2012}\natexlab{c}.
\newblock \showarticletitle{Don't Thrash: How to Cache Your Hash on Flash}.
\newblock \bibinfo{journal}{\emph{PVLDB}} \bibinfo{volume}{5},
  \bibinfo{number}{11} (\bibinfo{year}{2012}), \bibinfo{pages}{1627--1637}.
\newblock


\bibitem[Bloom(1970)]%
        {Bloom70}
\bibfield{author}{\bibinfo{person}{Burton~H. Bloom}.}
  \bibinfo{year}{1970}\natexlab{}.
\newblock \showarticletitle{Space/time Trade-offs in Hash Coding With Allowable
  Errors}.
\newblock \bibinfo{journal}{\emph{Commun. ACM}} \bibinfo{volume}{13},
  \bibinfo{number}{7} (\bibinfo{year}{1970}), \bibinfo{pages}{422--426}.
\newblock


\bibitem[Bonomi et~al\mbox{.}(2006)]%
        {BonomiMiPa06}
\bibfield{author}{\bibinfo{person}{Flavio Bonomi}, \bibinfo{person}{Michael
  Mitzenmacher}, \bibinfo{person}{Rina Panigrahy}, \bibinfo{person}{Sushil
  Singh}, {and} \bibinfo{person}{George Varghese}.}
  \bibinfo{year}{2006}\natexlab{}.
\newblock \showarticletitle{An improved construction for counting {B}loom
  filters}. In \bibinfo{booktitle}{\emph{European Symposium on Algorithms
  (ESA)}}. Springer, \bibinfo{pages}{684--695}.
\newblock


\bibitem[Bradley et~al\mbox{.}(2019)]%
        {bradley2019ultrafast}
\bibfield{author}{\bibinfo{person}{Phelim Bradley}, \bibinfo{person}{Henk~C
  Den~Bakker}, \bibinfo{person}{Eduardo~PC Rocha}, \bibinfo{person}{Gil
  McVean}, {and} \bibinfo{person}{Zamin Iqbal}.}
  \bibinfo{year}{2019}\natexlab{}.
\newblock \showarticletitle{Ultrafast search of all deposited bacterial and
  viral genomic data}.
\newblock \bibinfo{journal}{\emph{Nature biotechnology}} \bibinfo{volume}{37},
  \bibinfo{number}{2} (\bibinfo{year}{2019}), \bibinfo{pages}{152--159}.
\newblock


\bibitem[Breslow and Jayasena(2018)]%
        {BreslowJ18}
\bibfield{author}{\bibinfo{person}{Alex~D Breslow} {and}
  \bibinfo{person}{Nuwan~S Jayasena}.} \bibinfo{year}{2018}\natexlab{}.
\newblock \showarticletitle{Morton filters: faster, space-efficient cuckoo
  filters via biasing, compression, and decoupled logical sparsity}.
\newblock \bibinfo{journal}{\emph{Proceedings of the VLDB Endowment}}
  \bibinfo{volume}{11}, \bibinfo{number}{9} (\bibinfo{year}{2018}),
  \bibinfo{pages}{1041--1055}.
\newblock


\bibitem[Brodal and Fagerberg(2003)]%
        {BrodalFa03b}
\bibfield{author}{\bibinfo{person}{Gerth~St{\o}lting Brodal} {and}
  \bibinfo{person}{Rolf Fagerberg}.} \bibinfo{year}{2003}\natexlab{}.
\newblock \showarticletitle{Lower Bounds for External Memory Dictionaries}. In
  \bibinfo{booktitle}{\emph{Proceedings of the 14th Annual ACM-SIAM Symposium
  on Discrete Algorithms (SODA)}}. \bibinfo{pages}{546--554}.
\newblock


\bibitem[Broder and Mitzenmacher(2004)]%
        {BroderMi04}
\bibfield{author}{\bibinfo{person}{Andrei Broder} {and}
  \bibinfo{person}{Michael Mitzenmacher}.} \bibinfo{year}{2004}\natexlab{}.
\newblock \showarticletitle{Network applications of {B}loom filters: A survey}.
\newblock \bibinfo{journal}{\emph{Internet Mathematics}} \bibinfo{volume}{1},
  \bibinfo{number}{4} (\bibinfo{year}{2004}), \bibinfo{pages}{485--509}.
\newblock


\bibitem[CAIDA(2016)]%
        {caida2016}
\bibfield{author}{\bibinfo{person}{CAIDA}.} \bibinfo{year}{2016}\natexlab{}.
\newblock \bibinfo{title}{Anonymized Internet Traces 2016}.
\newblock
\newblock
\urldef\tempurl%
\url{https://www.caida.org/catalog/datasets/passive_dataset/}
\showURL{%
\tempurl}


\bibitem[Canim et~al\mbox{.}(2010)]%
        {CanimMiBh10}
\bibfield{author}{\bibinfo{person}{Mustafa Canim}, \bibinfo{person}{George~A
  Mihaila}, \bibinfo{person}{Bishwaranjan Bhattacharjee},
  \bibinfo{person}{Christian~A Lang}, {and} \bibinfo{person}{Kenneth~A Ross}.}
  \bibinfo{year}{2010}\natexlab{}.
\newblock \showarticletitle{Buffered {Bloom} Filters on Solid State Storage.}.
  In \bibinfo{booktitle}{\emph{Proceedings of the International Workshop on
  Accelerating Analytics and Data Management Systems Using Modern Processor and
  Storage Architectures (ADMS)}}. \bibinfo{pages}{1--8}.
\newblock


\bibitem[Carter et~al\mbox{.}(1978)]%
        {CarterFG78}
\bibfield{author}{\bibinfo{person}{Larry Carter}, \bibinfo{person}{Robert
  Floyd}, \bibinfo{person}{John Gill}, \bibinfo{person}{George Markowsky},
  {and} \bibinfo{person}{Mark Wegman}.} \bibinfo{year}{1978}\natexlab{}.
\newblock \showarticletitle{Exact and approximate membership testers}. In
  \bibinfo{booktitle}{\emph{Proceedings of the tenth annual ACM symposium on
  Theory of computing}}. \bibinfo{pages}{59--65}.
\newblock


\bibitem[Celis et~al\mbox{.}(1985)]%
        {CelisLaMu85}
\bibfield{author}{\bibinfo{person}{Pedro Celis}, \bibinfo{person}{Per-Ake
  Larson}, {and} \bibinfo{person}{J~Ian Munro}.}
  \bibinfo{year}{1985}\natexlab{}.
\newblock \showarticletitle{Robin hood hashing}. In
  \bibinfo{booktitle}{\emph{26th Annual Symposium on Foundations of Computer
  Science (FOCS)}}. \bibinfo{pages}{281--288}.
\newblock


\bibitem[Chang et~al\mbox{.}(2006)]%
        {ChangDG06}
\bibfield{author}{\bibinfo{person}{Fay Chang}, \bibinfo{person}{Jeffrey Dean},
  \bibinfo{person}{Sanjay Ghemawat}, \bibinfo{person}{Wilson~C. Hsieh},
  \bibinfo{person}{Deborah~A. Wallach}, \bibinfo{person}{Mike Burrows},
  \bibinfo{person}{Tushar Chandra}, \bibinfo{person}{Andrew Fikes}, {and}
  \bibinfo{person}{Robert~E. Gruber}.} \bibinfo{year}{2006}\natexlab{}.
\newblock \showarticletitle{Bigtable: A Distributed Storage System for
  Structured Data}. In \bibinfo{booktitle}{\emph{7th {USENIX} Symposium on
  Operating Systems Design and Implementation (OSDI)}}.
  \bibinfo{pages}{205--218}.
\newblock


\bibitem[Chazelle et~al\mbox{.}(2004)]%
        {chazelle2004bloomier}
\bibfield{author}{\bibinfo{person}{Bernard Chazelle}, \bibinfo{person}{Joe
  Kilian}, \bibinfo{person}{Ronitt Rubinfeld}, {and} \bibinfo{person}{Ayellet
  Tal}.} \bibinfo{year}{2004}\natexlab{}.
\newblock \showarticletitle{The {Bloomier} filter: an efficient data structure
  for static support lookup tables}. In \bibinfo{booktitle}{\emph{Proceedings
  of the fifteenth annual ACM-SIAM symposium on Discrete algorithms}}. Society
  for Industrial and Applied Mathematics, \bibinfo{pages}{30--39}.
\newblock


\bibitem[Chen et~al\mbox{.}(2011)]%
        {ChenAA11}
\bibfield{author}{\bibinfo{person}{Shimin Chen}, \bibinfo{person}{Anastasia
  Ailamaki}, \bibinfo{person}{Manos Athanassoulis}, \bibinfo{person}{Phillip~B
  Gibbons}, \bibinfo{person}{Ryan Johnson}, \bibinfo{person}{Ippokratis
  Pandis}, {and} \bibinfo{person}{Radu Stoica}.}
  \bibinfo{year}{2011}\natexlab{}.
\newblock \showarticletitle{TPC-E vs. TPC-C: Characterizing the new TPC-E
  benchmark via an I/O comparison study}.
\newblock \bibinfo{journal}{\emph{ACM Sigmod Record}} \bibinfo{volume}{39},
  \bibinfo{number}{3} (\bibinfo{year}{2011}), \bibinfo{pages}{5--10}.
\newblock


\bibitem[Cheng et~al\mbox{.}(2014)]%
        {ChengKo14}
\bibfield{author}{\bibinfo{person}{Long Cheng}, \bibinfo{person}{Spyros
  Kotoulas}, \bibinfo{person}{Tomas~E. Ward}, {and} \bibinfo{person}{Georgios
  Theodoropoulos}.} \bibinfo{year}{2014}\natexlab{}.
\newblock \showarticletitle{Robust and Skew-Resistant Parallel Joins in
  Shared-Nothing Systems}. In \bibinfo{booktitle}{\emph{Proceedings of the 23rd
  ACM International Conference on Conference on Information and Knowledge
  Management}} (Shanghai, China) \emph{(\bibinfo{series}{CIKM '14})}.
  \bibinfo{publisher}{Association for Computing Machinery},
  \bibinfo{address}{New York, NY, USA}, \bibinfo{pages}{1399–1408}.
\newblock
\showISBNx{9781450325981}
\urldef\tempurl%
\url{https://doi.org/10.1145/2661829.2661888}
\showDOI{\tempurl}


\bibitem[Chikhi and Rizk(2013a)]%
        {ChikhiRi13}
\bibfield{author}{\bibinfo{person}{Rayan Chikhi} {and}
  \bibinfo{person}{Guillaume Rizk}.} \bibinfo{year}{2013}\natexlab{a}.
\newblock \showarticletitle{Space-efficient and exact de {Bruijn} graph
  representation based on a {B}loom filter}.
\newblock \bibinfo{journal}{\emph{Algorithms for Molecular Biology}}
  \bibinfo{volume}{8}, \bibinfo{number}{1} (\bibinfo{year}{2013}),
  \bibinfo{pages}{1}.
\newblock


\bibitem[Chikhi and Rizk(2013b)]%
        {chikhi2013space}
\bibfield{author}{\bibinfo{person}{Rayan Chikhi} {and}
  \bibinfo{person}{Guillaume Rizk}.} \bibinfo{year}{2013}\natexlab{b}.
\newblock \showarticletitle{Space-efficient and exact de {Bruijn} graph
  representation based on a {Bloom} filter}.
\newblock \bibinfo{journal}{\emph{Algorithms for Molecular Biology}}
  \bibinfo{volume}{8}, \bibinfo{number}{1} (\bibinfo{year}{2013}),
  \bibinfo{pages}{22}.
\newblock


\bibitem[Chu et~al\mbox{.}(2014)]%
        {chu2014biobloom}
\bibfield{author}{\bibinfo{person}{Justin Chu}, \bibinfo{person}{Sara Sadeghi},
  \bibinfo{person}{Anthony Raymond}, \bibinfo{person}{Shaun~D Jackman},
  \bibinfo{person}{Ka~Ming Nip}, \bibinfo{person}{Richard Mar},
  \bibinfo{person}{Hamid Mohamadi}, \bibinfo{person}{Yaron~S Butterfield},
  \bibinfo{person}{A~Gordon Robertson}, {and} \bibinfo{person}{Inanc Birol}.}
  \bibinfo{year}{2014}\natexlab{}.
\newblock \showarticletitle{BioBloom tools: fast, accurate and memory-efficient
  host species sequence screening using bloom filters}.
\newblock \bibinfo{journal}{\emph{Bioinformatics}} \bibinfo{volume}{30},
  \bibinfo{number}{23} (\bibinfo{year}{2014}), \bibinfo{pages}{3402--3404}.
\newblock


\bibitem[Clauset et~al\mbox{.}(2009)]%
        {ClausetSh09}
\bibfield{author}{\bibinfo{person}{Aaron Clauset},
  \bibinfo{person}{Cosma~Rohilla Shalizi}, {and} \bibinfo{person}{Mark~EJ
  Newman}.} \bibinfo{year}{2009}\natexlab{}.
\newblock \showarticletitle{Power-law distributions in empirical data}.
\newblock \bibinfo{journal}{\emph{SIAM review}} \bibinfo{volume}{51},
  \bibinfo{number}{4} (\bibinfo{year}{2009}), \bibinfo{pages}{661--703}.
\newblock


\bibitem[Conway et~al\mbox{.}(2020)]%
        {ConwayGC20}
\bibfield{author}{\bibinfo{person}{Alexander Conway}, \bibinfo{person}{Abhishek
  Gupta}, \bibinfo{person}{Vijay Chidambaram}, \bibinfo{person}{Martin
  Farach-Colton}, \bibinfo{person}{Richard Spillane}, \bibinfo{person}{Amy
  Tai}, {and} \bibinfo{person}{Rob Johnson}.} \bibinfo{year}{2020}\natexlab{}.
\newblock \showarticletitle{$\{$SplinterDB$\}$: Closing the Bandwidth Gap for
  $\{$NVMe$\}$$\{$Key-Value$\}$ Stores}. In \bibinfo{booktitle}{\emph{2020
  USENIX Annual Technical Conference (USENIX ATC 20)}}.
  \bibinfo{pages}{49--63}.
\newblock


\bibitem[Cooper et~al\mbox{.}(2010)]%
        {CooperSi10}
\bibfield{author}{\bibinfo{person}{Brian~F Cooper}, \bibinfo{person}{Adam
  Silberstein}, \bibinfo{person}{Erwin Tam}, \bibinfo{person}{Raghu
  Ramakrishnan}, {and} \bibinfo{person}{Russell Sears}.}
  \bibinfo{year}{2010}\natexlab{}.
\newblock \showarticletitle{Benchmarking cloud serving systems with YCSB}. In
  \bibinfo{booktitle}{\emph{Proceedings of the 1st ACM symposium on Cloud
  computing}}. \bibinfo{pages}{143--154}.
\newblock


\bibitem[Corominas-Murtra and Sol{\'e}(2010)]%
        {corominasSo10}
\bibfield{author}{\bibinfo{person}{Bernat Corominas-Murtra} {and}
  \bibinfo{person}{Ricard~V Sol{\'e}}.} \bibinfo{year}{2010}\natexlab{}.
\newblock \showarticletitle{Universality of {Z}ipf's law}.
\newblock \bibinfo{journal}{\emph{Physical Review E}} \bibinfo{volume}{82},
  \bibinfo{number}{1} (\bibinfo{year}{2010}), \bibinfo{pages}{011102}.
\newblock


\bibitem[Dayan et~al\mbox{.}(2017)]%
        {dayan2017monkey}
\bibfield{author}{\bibinfo{person}{Niv Dayan}, \bibinfo{person}{Manos
  Athanassoulis}, {and} \bibinfo{person}{Stratos Idreos}.}
  \bibinfo{year}{2017}\natexlab{}.
\newblock \showarticletitle{Monkey: Optimal navigable key-value store}. In
  \bibinfo{booktitle}{\emph{Proceedings of the 2017 ACM International
  Conference on Management of Data}}. \bibinfo{pages}{79--94}.
\newblock


\bibitem[Debnath et~al\mbox{.}(2011)]%
        {DebnathSeLi11}
\bibfield{author}{\bibinfo{person}{Biplob Debnath}, \bibinfo{person}{Sudipta
  Sengupta}, \bibinfo{person}{Jin Li}, \bibinfo{person}{David~J Lilja}, {and}
  \bibinfo{person}{David~HC Du}.} \bibinfo{year}{2011}\natexlab{}.
\newblock \showarticletitle{{BloomFlash}: {B}loom filter on flash-based
  storage}. In \bibinfo{booktitle}{\emph{Proceedings of the 31st International
  Conference on Distributed Computing Systems (ICDCS)}}.
  \bibinfo{pages}{635--644}.
\newblock


\bibitem[Debnath et~al\mbox{.}(2010)]%
        {DebnathDeLi10}
\bibfield{author}{\bibinfo{person}{Biplob~K Debnath}, \bibinfo{person}{Sudipta
  Sengupta}, {and} \bibinfo{person}{Jin Li}.} \bibinfo{year}{2010}\natexlab{}.
\newblock \showarticletitle{ChunkStash: Speeding Up Inline Storage
  Deduplication Using Flash Memory.}. In \bibinfo{booktitle}{\emph{Proceedings
  of the USENIX Annual Technical Conference (ATC)}}.
\newblock


\bibitem[Deeds et~al\mbox{.}(2020)]%
        {deeds2020stacked}
\bibfield{author}{\bibinfo{person}{Kyle Deeds}, \bibinfo{person}{Brian
  Hentschel}, {and} \bibinfo{person}{Stratos Idreos}.}
  \bibinfo{year}{2020}\natexlab{}.
\newblock \showarticletitle{Stacked filters: learning to filter by structure}.
\newblock \bibinfo{journal}{\emph{Proceedings of the VLDB Endowment}}
  \bibinfo{volume}{14}, \bibinfo{number}{4} (\bibinfo{year}{2020}),
  \bibinfo{pages}{600--612}.
\newblock


\bibitem[Dillinger and Manolios(2009)]%
        {DillingerMa09}
\bibfield{author}{\bibinfo{person}{Peter~C. Dillinger} {and}
  \bibinfo{person}{Panagiotis~(Pete) Manolios}.}
  \bibinfo{year}{2009}\natexlab{}.
\newblock \showarticletitle{Fast, All-Purpose State Storage}. In
  \bibinfo{booktitle}{\emph{Proceedings of the 16th International SPIN Workshop
  on Model Checking Software}} (Grenoble, France).
  \bibinfo{publisher}{Springer-Verlag}, \bibinfo{address}{Berlin, Heidelberg},
  \bibinfo{pages}{12--31}.
\newblock
\showISBNx{9783642026515}
\urldef\tempurl%
\url{https://doi.org/10.1007/978-3-642-02652-2_6}
\showDOI{\tempurl}


\bibitem[Duggan et~al\mbox{.}(2015)]%
        {DugganPa15}
\bibfield{author}{\bibinfo{person}{Jennie Duggan}, \bibinfo{person}{Olga
  Papaemmanouil}, \bibinfo{person}{Leilani Battle}, {and}
  \bibinfo{person}{Michael Stonebraker}.} \bibinfo{year}{2015}\natexlab{}.
\newblock \showarticletitle{Skew-Aware Join Optimization for Array Databases}.
  In \bibinfo{booktitle}{\emph{Proceedings of the 2015 ACM SIGMOD International
  Conference on Management of Data}} (Melbourne, Victoria, Australia)
  \emph{(\bibinfo{series}{SIGMOD '15})}. \bibinfo{publisher}{Association for
  Computing Machinery}, \bibinfo{address}{New York, NY, USA},
  \bibinfo{pages}{123–135}.
\newblock
\showISBNx{9781450327589}
\urldef\tempurl%
\url{https://doi.org/10.1145/2723372.2723709}
\showDOI{\tempurl}


\bibitem[Einziger and Friedman(2016)]%
        {EinzigerFr16}
\bibfield{author}{\bibinfo{person}{Gil Einziger} {and} \bibinfo{person}{Roy
  Friedman}.} \bibinfo{year}{2016}\natexlab{}.
\newblock \showarticletitle{Counting with TinyTable: Every Bit Counts!}. In
  \bibinfo{booktitle}{\emph{Proceedings of the 17th International Conference on
  Distributed Computing and Networking}} (Singapore, Singapore)
  \emph{(\bibinfo{series}{ICDCN '16})}. \bibinfo{publisher}{Association for
  Computing Machinery}, \bibinfo{address}{New York, NY, USA}, Article
  \bibinfo{articleno}{27}, \bibinfo{numpages}{10}~pages.
\newblock
\showISBNx{9781450340328}
\urldef\tempurl%
\url{https://doi.org/10.1145/2833312.2833449}
\showDOI{\tempurl}


\bibitem[Eisenberg(2008)]%
        {Eisenberg2008MaxGeometric}
\bibfield{author}{\bibinfo{person}{Bennett Eisenberg}.}
  \bibinfo{year}{2008}\natexlab{}.
\newblock \showarticletitle{On the expectation of the maximum of IID geometric
  random variables}.
\newblock \bibinfo{journal}{\emph{Statistics \& Probability Letters}}
  \bibinfo{volume}{78}, \bibinfo{number}{2} (\bibinfo{year}{2008}),
  \bibinfo{pages}{135--143}.
\newblock
\urldef\tempurl%
\url{https://www.sciencedirect.com/science/article/pii/S0167715207002040}
\showURL{%
\tempurl}


\bibitem[Esmet et~al\mbox{.}(2012)]%
        {EsmetBeFa12}
\bibfield{author}{\bibinfo{person}{John Esmet}, \bibinfo{person}{Michael~A.
  Bender}, \bibinfo{person}{Martin Farach-Colton}, {and}
  \bibinfo{person}{Bradley~C. Kuszmaul}.} \bibinfo{year}{2012}\natexlab{}.
\newblock \showarticletitle{The {TokuFS} Streaming File System}. In
  \bibinfo{booktitle}{\emph{Proc.\ 4th USENIX Workshop on Hot Topics in Storage
  (HotStorage)}}. \bibinfo{address}{Boston, MA, USA}.
\newblock


\bibitem[Fan et~al\mbox{.}(2014a)]%
        {FanAnKa14}
\bibfield{author}{\bibinfo{person}{Bin Fan}, \bibinfo{person}{Dave~G Andersen},
  \bibinfo{person}{Michael Kaminsky}, {and} \bibinfo{person}{Michael~D
  Mitzenmacher}.} \bibinfo{year}{2014}\natexlab{a}.
\newblock \showarticletitle{Cuckoo Filter: Practically Better Than {Bloom}}. In
  \bibinfo{booktitle}{\emph{Proceedings of the 10th ACM International on
  Conference on Emerging Networking Experiments and Technologies}}.
  \bibinfo{pages}{75--88}.
\newblock


\bibitem[Fan et~al\mbox{.}(2014b)]%
        {fan2014cuckoo}
\bibfield{author}{\bibinfo{person}{Bin Fan}, \bibinfo{person}{Dave~G Andersen},
  \bibinfo{person}{Michael Kaminsky}, {and} \bibinfo{person}{Michael~D
  Mitzenmacher}.} \bibinfo{year}{2014}\natexlab{b}.
\newblock \showarticletitle{Cuckoo Filter: Practically Better Than {Bloom}}. In
  \bibinfo{booktitle}{\emph{Proceedings of the 10th ACM International on
  Conference on emerging Networking Experiments and Technologies}}. ACM,
  \bibinfo{pages}{75--88}.
\newblock


\bibitem[Fan et~al\mbox{.}(2000)]%
        {FanCaAl00}
\bibfield{author}{\bibinfo{person}{Li Fan}, \bibinfo{person}{Pei Cao},
  \bibinfo{person}{Jussara Almeida}, {and} \bibinfo{person}{Andrei~Z Broder}.}
  \bibinfo{year}{2000}\natexlab{}.
\newblock \showarticletitle{Summary cache: {A} scalable wide-area web cache
  sharing protocol}.
\newblock \bibinfo{journal}{\emph{IEEE/ACM Transactions on Networking (TON)}}
  \bibinfo{volume}{8}, \bibinfo{number}{3} (\bibinfo{year}{2000}),
  \bibinfo{pages}{281--293}.
\newblock


\bibitem[Farach-Colton et~al\mbox{.}(2009)]%
        {Farach-ColtonFeMo09}
\bibfield{author}{\bibinfo{person}{Martin Farach-Colton},
  \bibinfo{person}{Rohan~J. Fernandes}, {and} \bibinfo{person}{Miguel~A.
  Mosteiro}.} \bibinfo{year}{2009}\natexlab{}.
\newblock \showarticletitle{Bootstrapping a hop-optimal network in the weak
  sensor model}.
\newblock \bibinfo{journal}{\emph{ACM Transactions on Algorithms}}
  \bibinfo{volume}{5}, \bibinfo{number}{4} (\bibinfo{year}{2009}).
\newblock


\bibitem[{Google, Inc.}(2015)]%
        {LevelDB}
\bibfield{author}{\bibinfo{person}{{Google, Inc.}}}
  \bibinfo{year}{2015}\natexlab{}.
\newblock \bibinfo{title}{{LevelDB: A fast and lightweight key/value database
  library by Google}}.
\newblock \bibinfo{howpublished}{\url{http://github.com/leveldb/}, Last
  Accessed May 16, 2015}.
\newblock


\bibitem[Graf and Lemire(2020)]%
        {GrafLe20}
\bibfield{author}{\bibinfo{person}{Thomas~Mueller Graf} {and}
  \bibinfo{person}{Daniel Lemire}.} \bibinfo{year}{2020}\natexlab{}.
\newblock \showarticletitle{Xor Filters: Faster and Smaller Than Bloom and
  Cuckoo Filters}.
\newblock \bibinfo{journal}{\emph{ACM J. Exp. Algorithmics}}
  \bibinfo{volume}{25}, Article \bibinfo{articleno}{1.5} (\bibinfo{date}{March}
  \bibinfo{year}{2020}), \bibinfo{numpages}{16}~pages.
\newblock
\showISSN{1084-6654}
\urldef\tempurl%
\url{https://doi.org/10.1145/3376122}
\showDOI{\tempurl}


\bibitem[Hogan(2009)]%
        {Hogan09}
\bibfield{author}{\bibinfo{person}{Trish Hogan}.}
  \bibinfo{year}{2009}\natexlab{}.
\newblock \showarticletitle{Overview of tpc benchmark e: The next generation of
  oltp benchmarks}. In \bibinfo{booktitle}{\emph{Technology Conference on
  Performance Evaluation and Benchmarking}}. Springer, \bibinfo{pages}{84--98}.
\newblock


\bibitem[Housley et~al\mbox{.}(1999)]%
        {housley1999internet}
\bibfield{author}{\bibinfo{person}{Russell Housley}, \bibinfo{person}{Warwick
  Ford}, \bibinfo{person}{William Polk}, {and} \bibinfo{person}{David Solo}.}
  \bibinfo{year}{1999}\natexlab{}.
\newblock \bibinfo{booktitle}{\emph{Internet X. 509 public key infrastructure
  certificate and CRL profile}}.
\newblock \bibinfo{type}{{T}echnical {R}eport}.
\newblock


\bibitem[InternetLiveStats.com(2022)]%
        {urlcount}
\bibfield{author}{\bibinfo{person}{InternetLiveStats.com}.}
  \bibinfo{year}{2022}\natexlab{}.
\newblock \bibinfo{booktitle}{\emph{Google search statistics}}.
\newblock
\urldef\tempurl%
\url{https://www.internetlivestats.com/google-search-statistics/}
\showURL{%
\tempurl}


\bibitem[Jackman et~al\mbox{.}(2017)]%
        {jackman2017abyss}
\bibfield{author}{\bibinfo{person}{Shaun~D Jackman},
  \bibinfo{person}{Benjamin~P Vandervalk}, \bibinfo{person}{Hamid Mohamadi},
  \bibinfo{person}{Justin Chu}, \bibinfo{person}{Sarah Yeo},
  \bibinfo{person}{S~Austin Hammond}, \bibinfo{person}{Golnaz Jahesh},
  \bibinfo{person}{Hamza Khan}, \bibinfo{person}{Lauren Coombe},
  \bibinfo{person}{Rene~L Warren}, {et~al\mbox{.}}}
  \bibinfo{year}{2017}\natexlab{}.
\newblock \showarticletitle{{ABySS} 2.0: resource-efficient assembly of large
  genomes using a {Bloom} filter}.
\newblock \bibinfo{journal}{\emph{Genome research}} \bibinfo{volume}{27},
  \bibinfo{number}{5} (\bibinfo{year}{2017}), \bibinfo{pages}{768--777}.
\newblock


\bibitem[Jannen et~al\mbox{.}(2015a)]%
        {JannenYuZh15a}
\bibfield{author}{\bibinfo{person}{William Jannen}, \bibinfo{person}{Jun Yuan},
  \bibinfo{person}{Yang Zhan}, \bibinfo{person}{Amogh Akshintala},
  \bibinfo{person}{John Esmet}, \bibinfo{person}{Yizheng Jiao},
  \bibinfo{person}{Ankur Mittal}, \bibinfo{person}{Prashant Pandey},
  \bibinfo{person}{Phaneendra Reddy}, \bibinfo{person}{Leif Walsh},
  \bibinfo{person}{Michael~A. Bender}, \bibinfo{person}{Martin
  Farach{-}Colton}, \bibinfo{person}{Rob Johnson}, \bibinfo{person}{Bradley~C.
  Kuszmaul}, {and} \bibinfo{person}{Donald~E. Porter}.}
  \bibinfo{year}{2015}\natexlab{a}.
\newblock \showarticletitle{{BetrFS:} {A} Right-Optimized Write-Optimized File
  System}. In \bibinfo{booktitle}{\emph{Proc.\ 13th {USENIX} Conference on File
  and Storage Technologies (FAST)}} (Santa Clara, CA, USA),
  \bibfield{editor}{\bibinfo{person}{Jiri Schindler} {and}
  \bibinfo{person}{Erez Zadok}} (Eds.). \bibinfo{pages}{301--315}.
\newblock


\bibitem[Jannen et~al\mbox{.}(2015b)]%
        {JannenYuZh15b}
\bibfield{author}{\bibinfo{person}{William Jannen}, \bibinfo{person}{Jun Yuan},
  \bibinfo{person}{Yang Zhan}, \bibinfo{person}{Amogh Akshintala},
  \bibinfo{person}{John Esmet}, \bibinfo{person}{Yizheng Jiao},
  \bibinfo{person}{Ankur Mittal}, \bibinfo{person}{Prashant Pandey},
  \bibinfo{person}{Phaneendra Reddy}, \bibinfo{person}{Leif Walsh},
  \bibinfo{person}{Michael~A. Bender}, \bibinfo{person}{Martin
  Farach{-}Colton}, \bibinfo{person}{Rob Johnson}, \bibinfo{person}{Bradley~C.
  Kuszmaul}, {and} \bibinfo{person}{Donald~E. Porter}.}
  \bibinfo{year}{2015}\natexlab{b}.
\newblock \showarticletitle{{BetrFS:} Write-Optimization in a Kernel File
  System}.
\newblock \bibinfo{journal}{\emph{Transactions on Storage---Special Issue on
  USENIX FAST 2015}} \bibinfo{volume}{11}, \bibinfo{number}{4}
  (\bibinfo{year}{2015}), \bibinfo{pages}{18:1--18:29}.
\newblock


\bibitem[KG({[n.\,d.]})]%
        {shalla2021}
\bibfield{author}{\bibinfo{person}{Shalla Secure~Services KG}.}
  \bibinfo{year}{[n.\,d.]}\natexlab{}.
\newblock \bibinfo{title}{Shalla’s Blacklists.}
\newblock
\newblock
\urldef\tempurl%
\url{http://www.shallalist.de/index.html}
\showURL{%
\tempurl}


\bibitem[Kopelowitz et~al\mbox{.}(2021)]%
        {kopelowitz2021support}
\bibfield{author}{\bibinfo{person}{Tsvi Kopelowitz}, \bibinfo{person}{Samuel
  McCauley}, {and} \bibinfo{person}{Ely Porat}.}
  \bibinfo{year}{2021}\natexlab{}.
\newblock \showarticletitle{Support optimality and adaptive cuckoo filters}. In
  \bibinfo{booktitle}{\emph{Workshop on Algorithms and Data Structures}}.
  Springer, \bibinfo{pages}{556--570}.
\newblock


\bibitem[Larisch et~al\mbox{.}(2017)]%
        {Larisch2017}
\bibfield{author}{\bibinfo{person}{James Larisch}, \bibinfo{person}{David
  Choffnes}, \bibinfo{person}{Dave Levin}, \bibinfo{person}{Bruce~M. Maggs},
  \bibinfo{person}{Alan Mislove}, {and} \bibinfo{person}{Christo Wilson}.}
  \bibinfo{year}{2017}\natexlab{}.
\newblock \showarticletitle{CRLite: A Scalable System for Pushing All TLS
  Revocations to All Browsers}. In \bibinfo{booktitle}{\emph{2017 IEEE
  Symposium on Security and Privacy (SP)}}. \bibinfo{pages}{539--556}.
\newblock
\urldef\tempurl%
\url{https://doi.org/10.1109/SP.2017.17}
\showDOI{\tempurl}


\bibitem[Lee et~al\mbox{.}(2021)]%
        {lee2021telescoping}
\bibfield{author}{\bibinfo{person}{David~J. Lee}, \bibinfo{person}{Samuel
  McCauley}, \bibinfo{person}{Shikha Singh}, {and} \bibinfo{person}{Max
  Stein}.} \bibinfo{year}{2021}\natexlab{}.
\newblock \showarticletitle{Telescoping Filter: {A} Practical Adaptive Filter}.
\newblock   \bibinfo{volume}{204} (\bibinfo{year}{2021}),
  \bibinfo{pages}{60:1--60:18}.
\newblock
\urldef\tempurl%
\url{https://doi.org/10.4230/LIPIcs.ESA.2021.60}
\showDOI{\tempurl}


\bibitem[Leng(2022)]%
        {leng_2022}
\bibfield{author}{\bibinfo{person}{Allen Leng}.}
  \bibinfo{year}{2022}\natexlab{}.
\newblock \bibinfo{title}{1 in 2 visitors abandon a website that takes more
  than 6 seconds to load}.
\newblock
\newblock
\urldef\tempurl%
\url{https://digital.com/1-in-2-visitors-abandon-a-website-that-takes-more-than-6-seconds-to-load/}
\showURL{%
\tempurl}


\bibitem[Leutenegger and Dias(1993)]%
        {LeuteneggerDi93}
\bibfield{author}{\bibinfo{person}{Scott~T Leutenegger} {and}
  \bibinfo{person}{Daniel Dias}.} \bibinfo{year}{1993}\natexlab{}.
\newblock \showarticletitle{A modeling study of the TPC-C benchmark}.
\newblock \bibinfo{journal}{\emph{ACM Sigmod Record}} \bibinfo{volume}{22},
  \bibinfo{number}{2} (\bibinfo{year}{1993}), \bibinfo{pages}{22--31}.
\newblock


\bibitem[Li et~al\mbox{.}(2022)]%
        {seesaw2022}
\bibfield{author}{\bibinfo{person}{Meng Li}, \bibinfo{person}{Deyi Chen},
  \bibinfo{person}{Haipeng Dai}, \bibinfo{person}{Rongbiao Xie},
  \bibinfo{person}{Siqiang Luo}, \bibinfo{person}{Rong Gu},
  \bibinfo{person}{Tong Yang}, {and} \bibinfo{person}{Guihai Chen}.}
  \bibinfo{year}{2022}\natexlab{}.
\newblock \showarticletitle{Seesaw Counting Filter: An Efficient Guardian for
  Vulnerable Negative Keys During Dynamic Filtering}. In
  \bibinfo{booktitle}{\emph{Proceedings of the ACM Web Conference 2022}}
  (Virtual Event, Lyon, France) \emph{(\bibinfo{series}{WWW '22})}.
  \bibinfo{publisher}{Association for Computing Machinery},
  \bibinfo{address}{New York, NY, USA}, \bibinfo{pages}{2759–2767}.
\newblock
\showISBNx{9781450390965}
\urldef\tempurl%
\url{https://doi.org/10.1145/3485447.3511996}
\showDOI{\tempurl}


\bibitem[Lohr({[n.\,d.]})]%
        {nytwait2012}
\bibfield{author}{\bibinfo{person}{Steve Lohr}.}
  \bibinfo{year}{[n.\,d.]}\natexlab{}.
\newblock \showarticletitle{For Impatient Web Users, an Eye Blink Is Just Too
  Long to Wait}.
\newblock \bibinfo{journal}{\emph{The New York Times}}
  (\bibinfo{year}{[n.\,d.]}).
\newblock
\urldef\tempurl%
\url{https://www.nytimes.com/2012/03/01/technology/impatient-web-users-flee-slow-loading-sites.html}
\showURL{%
\tempurl}


\bibitem[Lu et~al\mbox{.}(2011)]%
        {LuDeDu11}
\bibfield{author}{\bibinfo{person}{Guanlin Lu}, \bibinfo{person}{Biplob
  Debnath}, {and} \bibinfo{person}{David~HC Du}.}
  \bibinfo{year}{2011}\natexlab{}.
\newblock \showarticletitle{A Forest-structured {Bloom} Filter with flash
  memory}. In \bibinfo{booktitle}{\emph{Proceedings of the 27th Symposium on
  Mass Storage Systems and Technologies (MSST)}}. \bibinfo{pages}{1--6}.
\newblock


\bibitem[Mitzenmacher et~al\mbox{.}(2020)]%
        {MitzenmacherPR20}
\bibfield{author}{\bibinfo{person}{Michael Mitzenmacher},
  \bibinfo{person}{Salvatore Pontarelli}, {and} \bibinfo{person}{Pedro
  Reviriego}.} \bibinfo{year}{2020}\natexlab{}.
\newblock \showarticletitle{Adaptive Cuckoo Filters}.
\newblock \bibinfo{journal}{\emph{{ACM} J. Exp. Algorithmics}}
  \bibinfo{volume}{25} (\bibinfo{year}{2020}), \bibinfo{pages}{1--20}.
\newblock
\urldef\tempurl%
\url{https://doi.org/10.1145/3339504}
\showDOI{\tempurl}


\bibitem[MongoDB({[n.\,d.]})]%
        {WiredTiger}
\bibfield{author}{\bibinfo{person}{MongoDB}.}
  \bibinfo{year}{[n.\,d.]}\natexlab{}.
\newblock \bibinfo{title}{WiredTiger}.
\newblock
  \bibinfo{howpublished}{\url{https://github.com/wiredtiger/wiredtiger}}.
\newblock


\bibitem[Mousavi and Tripunitara(2019)]%
        {MousaviTr19}
\bibfield{author}{\bibinfo{person}{Nima Mousavi} {and} \bibinfo{person}{Mahesh
  Tripunitara}.} \bibinfo{year}{2019}\natexlab{}.
\newblock \showarticletitle{Constructing cascade bloom filters for efficient
  access enforcement}.
\newblock \bibinfo{journal}{\emph{Computers \& Security}}  \bibinfo{volume}{81}
  (\bibinfo{year}{2019}), \bibinfo{pages}{1--14}.
\newblock
\showISSN{0167-4048}
\urldef\tempurl%
\url{https://doi.org/10.1016/j.cose.2018.09.015}
\showDOI{\tempurl}


\bibitem[Newman(2005)]%
        {Newman05}
\bibfield{author}{\bibinfo{person}{Mark~EJ Newman}.}
  \bibinfo{year}{2005}\natexlab{}.
\newblock \showarticletitle{Power laws, Pareto distributions and Zipf's law}.
\newblock \bibinfo{journal}{\emph{Contemporary physics}} \bibinfo{volume}{46},
  \bibinfo{number}{5} (\bibinfo{year}{2005}), \bibinfo{pages}{323--351}.
\newblock


\bibitem[Ngo et~al\mbox{.}(2014)]%
        {NgoRe14}
\bibfield{author}{\bibinfo{person}{Hung~Q Ngo}, \bibinfo{person}{Christopher
  R\'{e}}, {and} \bibinfo{person}{Atri Rudra}.}
  \bibinfo{year}{2014}\natexlab{}.
\newblock \showarticletitle{Skew Strikes Back: New Developments in the Theory
  of Join Algorithms}.
\newblock \bibinfo{journal}{\emph{SIGMOD Rec.}} \bibinfo{volume}{42},
  \bibinfo{number}{4} (\bibinfo{date}{feb} \bibinfo{year}{2014}),
  \bibinfo{pages}{5–16}.
\newblock
\showISSN{0163-5808}
\urldef\tempurl%
\url{https://doi.org/10.1145/2590989.2590991}
\showDOI{\tempurl}


\bibitem[O'Neil et~al\mbox{.}(1996)]%
        {OneilChGa96}
\bibfield{author}{\bibinfo{person}{Patrick O'Neil}, \bibinfo{person}{Edward
  Cheng}, \bibinfo{person}{Dieter Gawlic}, {and} \bibinfo{person}{Elizabeth
  O'Neil}.} \bibinfo{year}{1996}\natexlab{}.
\newblock \showarticletitle{The Log-Structured Merge-Tree ({LSM}-tree)}.
\newblock \bibinfo{journal}{\emph{Acta Informatica}} \bibinfo{volume}{33},
  \bibinfo{number}{4} (\bibinfo{year}{1996}), \bibinfo{pages}{351--385}.
\newblock
\urldef\tempurl%
\url{https://doi.org/10.1007/s002360050048}
\showDOI{\tempurl}


\bibitem[Pagh et~al\mbox{.}(2005)]%
        {PaghPR05}
\bibfield{author}{\bibinfo{person}{Anna Pagh}, \bibinfo{person}{Rasmus Pagh},
  {and} \bibinfo{person}{S~Srinivasa Rao}.} \bibinfo{year}{2005}\natexlab{}.
\newblock \showarticletitle{An optimal {Bloom} filter replacement}. In
  \bibinfo{booktitle}{\emph{Proceedings of the sixteenth annual ACM-SIAM
  symposium on Discrete algorithms}}. Society for Industrial and Applied
  Mathematics, \bibinfo{pages}{823--829}.
\newblock


\bibitem[Pandey et~al\mbox{.}(2018)]%
        {PandeyABFJP18Cell}
\bibfield{author}{\bibinfo{person}{Prashant Pandey}, \bibinfo{person}{Fatemeh
  Almodaresi}, \bibinfo{person}{Michael~A Bender}, \bibinfo{person}{Michael
  Ferdman}, \bibinfo{person}{Rob Johnson}, {and} \bibinfo{person}{Rob Patro}.}
  \bibinfo{year}{2018}\natexlab{}.
\newblock \showarticletitle{Mantis: A fast, small, and exact large-scale
  sequence-search index}.
\newblock \bibinfo{journal}{\emph{Cell systems}} \bibinfo{volume}{7},
  \bibinfo{number}{2} (\bibinfo{year}{2018}), \bibinfo{pages}{201--207}.
\newblock


\bibitem[Pandey et~al\mbox{.}(2017a)]%
        {PandeyBJP17a}
\bibfield{author}{\bibinfo{person}{Prashant Pandey}, \bibinfo{person}{Michael~A
  Bender}, \bibinfo{person}{Rob Johnson}, {and} \bibinfo{person}{Rob Patro}.}
  \bibinfo{year}{2017}\natexlab{a}.
\newblock \showarticletitle{{deBGR}: an efficient and near-exact representation
  of the weighted de {B}ruijn graph}.
\newblock \bibinfo{journal}{\emph{Bioinformatics}} \bibinfo{volume}{33},
  \bibinfo{number}{14} (\bibinfo{year}{2017}), \bibinfo{pages}{i133--i141}.
\newblock


\bibitem[Pandey et~al\mbox{.}(2017b)]%
        {PandeyBJP17}
\bibfield{author}{\bibinfo{person}{Prashant Pandey}, \bibinfo{person}{Michael~A
  Bender}, \bibinfo{person}{Rob Johnson}, {and} \bibinfo{person}{Rob Patro}.}
  \bibinfo{year}{2017}\natexlab{b}.
\newblock \showarticletitle{A general-purpose counting filter: Making every bit
  count}. In \bibinfo{booktitle}{\emph{Proceedings of the 2017 {ACM}
  International Conference on Management of Data}}. \bibinfo{pages}{775--787}.
\newblock


\bibitem[Pandey et~al\mbox{.}(2017c)]%
        {PandeyBJP17b}
\bibfield{author}{\bibinfo{person}{Prashant Pandey}, \bibinfo{person}{Michael~A
  Bender}, \bibinfo{person}{Rob Johnson}, {and} \bibinfo{person}{Rob Patro}.}
  \bibinfo{year}{2017}\natexlab{c}.
\newblock \showarticletitle{Squeakr: an exact and approximate k-mer counting
  system}.
\newblock \bibinfo{journal}{\emph{Bioinformatics}} \bibinfo{volume}{34},
  \bibinfo{number}{4} (\bibinfo{year}{2017}), \bibinfo{pages}{568--575}.
\newblock


\bibitem[Pandey et~al\mbox{.}(2021)]%
        {pandey2021vector}
\bibfield{author}{\bibinfo{person}{Prashant Pandey}, \bibinfo{person}{Alex
  Conway}, \bibinfo{person}{Joe Durie}, \bibinfo{person}{Michael~A Bender},
  \bibinfo{person}{Martin Farach-Colton}, {and} \bibinfo{person}{Rob Johnson}.}
  \bibinfo{year}{2021}\natexlab{}.
\newblock \showarticletitle{Vector quotient filters: Overcoming the time/space
  trade-off in filter design}. In \bibinfo{booktitle}{\emph{Proceedings of the
  2021 International Conference on Management of Data}}.
  \bibinfo{pages}{1386--1399}.
\newblock


\bibitem[Pell et~al\mbox{.}(2012)]%
        {pell2012scaling}
\bibfield{author}{\bibinfo{person}{Jason Pell}, \bibinfo{person}{Arend Hintze},
  \bibinfo{person}{Rosangela Canino-Koning}, \bibinfo{person}{Adina Howe},
  \bibinfo{person}{James~M Tiedje}, {and} \bibinfo{person}{C~Titus Brown}.}
  \bibinfo{year}{2012}\natexlab{}.
\newblock \showarticletitle{Scaling metagenome sequence assembly with
  probabilistic de {Bruijn} graphs}.
\newblock \bibinfo{journal}{\emph{Proceedings of the National Academy of
  Sciences}} \bibinfo{volume}{109}, \bibinfo{number}{33}
  (\bibinfo{year}{2012}), \bibinfo{pages}{13272--13277}.
\newblock


\bibitem[Putze et~al\mbox{.}(2007)]%
        {PutzeSaSi07}
\bibfield{author}{\bibinfo{person}{Felix Putze}, \bibinfo{person}{Peter
  Sanders}, {and} \bibinfo{person}{Johannes Singler}.}
  \bibinfo{year}{2007}\natexlab{}.
\newblock \showarticletitle{Cache-, hash-and space-efficient bloom filters}. In
  \bibinfo{booktitle}{\emph{International Workshop on Experimental and
  Efficient Algorithms}}. \bibinfo{pages}{108--121}.
\newblock


\bibitem[Qiao et~al\mbox{.}(2014)]%
        {QiaoLiCh14}
\bibfield{author}{\bibinfo{person}{Yan Qiao}, \bibinfo{person}{Tao Li}, {and}
  \bibinfo{person}{Shigang Chen}.} \bibinfo{year}{2014}\natexlab{}.
\newblock \showarticletitle{Fast {Bloom} Filters and Their Generalization}.
\newblock \bibinfo{journal}{\emph{IEEE Transactions on Parallel and Distributed
  Systems (TPDS)}} \bibinfo{volume}{25}, \bibinfo{number}{1}
  (\bibinfo{year}{2014}), \bibinfo{pages}{93--103}.
\newblock


\bibitem[Reagen et~al\mbox{.}(2017)]%
        {ReagenGA17}
\bibfield{author}{\bibinfo{person}{Brandon Reagen}, \bibinfo{person}{Udit
  Gupta}, \bibinfo{person}{Robert Adolf}, \bibinfo{person}{Michael~M
  Mitzenmacher}, \bibinfo{person}{Alexander~M Rush}, \bibinfo{person}{Gu-Yeon
  Wei}, {and} \bibinfo{person}{David Brooks}.} \bibinfo{year}{2017}\natexlab{}.
\newblock \showarticletitle{Weightless: Lossy weight encoding for deep neural
  network compression}.
\newblock \bibinfo{journal}{\emph{arXiv preprint arXiv:1711.04686}}
  (\bibinfo{year}{2017}).
\newblock


\bibitem[Reviriego et~al\mbox{.}(2021a)]%
        {Reviriego2021BlacklistFilter}
\bibfield{author}{\bibinfo{person}{Pedro Reviriego}, \bibinfo{person}{Alfonso
  S{\'a}nchez-Maci{\'a}n}, \bibinfo{person}{Stefan Walzer}, {and}
  \bibinfo{person}{Peter~C. Dillinger}.} \bibinfo{year}{2021}\natexlab{a}.
\newblock \bibinfo{title}{Approximate Membership Query Filters with a False
  Positive Free Set}.
\newblock
\newblock


\bibitem[Reviriego et~al\mbox{.}(2021b)]%
        {reviriego2021}
\bibfield{author}{\bibinfo{person}{Pedro Reviriego}, \bibinfo{person}{Alfonso
  Sánchez-Macián}, \bibinfo{person}{Stefan Walzer}, {and}
  \bibinfo{person}{Peter~C. Dillinger}.} \bibinfo{year}{2021}\natexlab{b}.
\newblock \bibinfo{title}{Approximate Membership Query Filters with a False
  Positive Free Set}.
\newblock
\newblock
\urldef\tempurl%
\url{https://doi.org/10.48550/ARXIV.2111.06856}
\showDOI{\tempurl}


\bibitem[RocksDB({[n.\,d.]})]%
        {RocksDB}
RocksDB \bibinfo{year}{[n.\,d.]}\natexlab{}.
\newblock \bibinfo{title}{{R}ocks{DB}}.
\newblock \bibinfo{howpublished}{\url{https://rocksdb.org/}, Last Accessed Oct.
  15, 2022}.
\newblock


\bibitem[Rozov et~al\mbox{.}(2014)]%
        {RozovShHa14}
\bibfield{author}{\bibinfo{person}{Roye Rozov}, \bibinfo{person}{Ron Shamir},
  {and} \bibinfo{person}{Eran Halperin}.} \bibinfo{year}{2014}\natexlab{}.
\newblock \showarticletitle{Fast lossless compression via cascading Bloom
  filters}.
\newblock \bibinfo{journal}{\emph{BMC Bioinformatics}}  \bibinfo{volume}{15}
  (\bibinfo{year}{2014}).
\newblock


\bibitem[Salikhov et~al\mbox{.}(2014)]%
        {salikhov2014using}
\bibfield{author}{\bibinfo{person}{Kamil Salikhov}, \bibinfo{person}{Gustavo
  Sacomoto}, {and} \bibinfo{person}{Gregory Kucherov}.}
  \bibinfo{year}{2014}\natexlab{}.
\newblock \showarticletitle{Using cascading Bloom filters to improve the memory
  usage for de Brujin graphs}.
\newblock \bibinfo{journal}{\emph{Algorithms for Molecular Biology}}
  \bibinfo{volume}{9}, \bibinfo{number}{1} (\bibinfo{year}{2014}),
  \bibinfo{pages}{1--10}.
\newblock


\bibitem[ScyllaDB({[n.\,d.]})]%
        {ScyllaDB}
\bibfield{author}{\bibinfo{person}{ScyllaDB}.}
  \bibinfo{year}{[n.\,d.]}\natexlab{}.
\newblock \bibinfo{title}{ScyllaDB}.
\newblock \bibinfo{howpublished}{\url{https://www.scylladb.com/}}.
\newblock


\bibitem[Securelist.com(2022)]%
        {kaspersky2021}
\bibfield{author}{\bibinfo{person}{Securelist.com}.}
  \bibinfo{year}{2022}\natexlab{}.
\newblock \bibinfo{booktitle}{}.
\newblock
\urldef\tempurl%
\url{https://securelist.com/kaspersky-security-bulletin-2021-statistics/105205/}
\showURL{%
\tempurl}


\bibitem[Solomon and Kingsford(2016)]%
        {solomon2016fast}
\bibfield{author}{\bibinfo{person}{Brad Solomon} {and} \bibinfo{person}{Carl
  Kingsford}.} \bibinfo{year}{2016}\natexlab{}.
\newblock \showarticletitle{Fast search of thousands of short-read sequencing
  experiments}.
\newblock \bibinfo{journal}{\emph{Nature biotechnology}} \bibinfo{volume}{34},
  \bibinfo{number}{3} (\bibinfo{year}{2016}), \bibinfo{pages}{300}.
\newblock


\bibitem[Stranneheim et~al\mbox{.}(2010)]%
        {stranneheim2010classification}
\bibfield{author}{\bibinfo{person}{Henrik Stranneheim}, \bibinfo{person}{Max
  K{\"a}ller}, \bibinfo{person}{Tobias Allander}, \bibinfo{person}{Bj{\"o}rn
  Andersson}, \bibinfo{person}{Lars Arvestad}, {and} \bibinfo{person}{Joakim
  Lundeberg}.} \bibinfo{year}{2010}\natexlab{}.
\newblock \showarticletitle{Classification of DNA sequences using Bloom
  filters}.
\newblock \bibinfo{journal}{\emph{Bioinformatics}} \bibinfo{volume}{26},
  \bibinfo{number}{13} (\bibinfo{year}{2010}), \bibinfo{pages}{1595--1600}.
\newblock


\bibitem[Sun et~al\mbox{.}(2016)]%
        {sun2016automating}
\bibfield{author}{\bibinfo{person}{Bo Sun}, \bibinfo{person}{Mitsuaki Akiyama},
  \bibinfo{person}{Takeshi Yagi}, \bibinfo{person}{Mitsuhiro Hatada}, {and}
  \bibinfo{person}{Tatsuya Mori}.} \bibinfo{year}{2016}\natexlab{}.
\newblock \showarticletitle{Automating URL blacklist generation with similarity
  search approach}.
\newblock \bibinfo{journal}{\emph{IEICE TRANSACTIONS on Information and
  Systems}} \bibinfo{volume}{99}, \bibinfo{number}{4} (\bibinfo{year}{2016}),
  \bibinfo{pages}{873--882}.
\newblock


\bibitem[Tripunitara and Carbunar(2009)]%
        {TripunitaraCa09}
\bibfield{author}{\bibinfo{person}{Mahesh~V. Tripunitara} {and}
  \bibinfo{person}{Bogdan Carbunar}.} \bibinfo{year}{2009}\natexlab{}.
\newblock \showarticletitle{Efficient Access Enforcement in Distributed
  Role-Based Access Control (RBAC) Deployments}. In
  \bibinfo{booktitle}{\emph{Proceedings of the 14th ACM Symposium on Access
  Control Models and Technologies}} (Stresa, Italy)
  \emph{(\bibinfo{series}{SACMAT '09})}. \bibinfo{publisher}{Association for
  Computing Machinery}, \bibinfo{address}{New York, NY, USA},
  \bibinfo{pages}{155–164}.
\newblock
\showISBNx{9781605585376}
\urldef\tempurl%
\url{https://doi.org/10.1145/1542207.1542232}
\showDOI{\tempurl}


\bibitem[Wajc(2017)]%
        {Wajc2017NA}
\bibfield{author}{\bibinfo{person}{David Wajc}.}
  \bibinfo{year}{2017}\natexlab{}.
\newblock \bibinfo{title}{Negative Association - Definition, Properties, and
  Applications}.
\newblock
\newblock


\bibitem[Wang et~al\mbox{.}(2014)]%
        {WangSuJi14}
\bibfield{author}{\bibinfo{person}{Peng Wang}, \bibinfo{person}{Guangyu Sun},
  \bibinfo{person}{Song Jiang}, \bibinfo{person}{Jian Ouyang},
  \bibinfo{person}{Shiding Lin}, \bibinfo{person}{Chen Zhang}, {and}
  \bibinfo{person}{Jason Cong}.} \bibinfo{year}{2014}\natexlab{}.
\newblock \showarticletitle{An efficient design and implementation of
  {LSM}-tree based key-value store on open-channel {SSD}}. In
  \bibinfo{booktitle}{\emph{Proceedings of the 9th European Conference on
  Computer Systems (EuroSys)}}. \bibinfo{pages}{16:1--16:14}.
\newblock


\bibitem[Yuan et~al\mbox{.}(2017)]%
        {yuan2017writes}
\bibfield{author}{\bibinfo{person}{Jun Yuan}, \bibinfo{person}{Yang Zhan},
  \bibinfo{person}{William Jannen}, \bibinfo{person}{Prashant Pandey},
  \bibinfo{person}{Amogh Akshintala}, \bibinfo{person}{Kanchan Chandnani},
  \bibinfo{person}{Pooja Deo}, \bibinfo{person}{Zardosht Kasheff},
  \bibinfo{person}{Leif Walsh}, \bibinfo{person}{Michael~A Bender},
  {et~al\mbox{.}}} \bibinfo{year}{2017}\natexlab{}.
\newblock \showarticletitle{Writes wrought right, and other adventures in file
  system optimization}.
\newblock \bibinfo{journal}{\emph{ACM Transactions on Storage (TOS)}}
  \bibinfo{volume}{13}, \bibinfo{number}{1} (\bibinfo{year}{2017}),
  \bibinfo{pages}{1--26}.
\newblock


\bibitem[Yuan et~al\mbox{.}(2016)]%
        {YuanZhJa16}
\bibfield{author}{\bibinfo{person}{Jun Yuan}, \bibinfo{person}{Yang Zhan},
  \bibinfo{person}{William Jannen}, \bibinfo{person}{Prashant Pandey},
  \bibinfo{person}{Amogh Akshintala}, \bibinfo{person}{Kanchan Chandnani},
  \bibinfo{person}{Pooja Deo}, \bibinfo{person}{Zardosht Kasheff},
  \bibinfo{person}{Leif Walsh}, \bibinfo{person}{Michael~A. Bender},
  \bibinfo{person}{Martin Farach-Colton}, \bibinfo{person}{Rob Johnson},
  \bibinfo{person}{Bradley~C. Kuszmaul}, {and} \bibinfo{person}{Donald~E.
  Porter}.} \bibinfo{year}{2016}\natexlab{}.
\newblock \showarticletitle{Optimizing Every Operation in a Write-Optimized
  File System}. In \bibinfo{booktitle}{\emph{Proc.\ 14th {USENIX} Conference on
  File and Storage Technologies (FAST)}}.
\newblock


\bibitem[Zhang and Ross(2022)]%
        {ZhangRo22}
\bibfield{author}{\bibinfo{person}{Wangda Zhang} {and}
  \bibinfo{person}{Kenneth~A. Ross}.} \bibinfo{year}{2022}\natexlab{}.
\newblock \showarticletitle{Exploiting Data Skew for Improved Query
  Performance}.
\newblock \bibinfo{journal}{\emph{IEEE Transactions on Knowledge and Data
  Engineering}} \bibinfo{volume}{34}, \bibinfo{number}{5}
  (\bibinfo{year}{2022}), \bibinfo{pages}{2176--2189}.
\newblock
\urldef\tempurl%
\url{https://doi.org/10.1109/TKDE.2020.3006446}
\showDOI{\tempurl}


\bibitem[Zhu et~al\mbox{.}(2008)]%
        {ZhuLiPa08}
\bibfield{author}{\bibinfo{person}{Benjamin Zhu}, \bibinfo{person}{Kai Li},
  {and} \bibinfo{person}{R~Hugo Patterson}.} \bibinfo{year}{2008}\natexlab{}.
\newblock \showarticletitle{Avoiding the Disk Bottleneck in the Data Domain
  Deduplication File System}. In \bibinfo{booktitle}{\emph{Proceedings of the
  6th USENIX Conference on File and Storage Technologies (FAST)}}.
  \bibinfo{pages}{1--14}.
\newblock


\end{thebibliography}

\end{document}